\documentclass[11pt,letterpaper]{article}

\usepackage[dvipsnames]{xcolor}
\usepackage{tikz}
\usetikzlibrary{angles,quotes,calc}
\usetikzlibrary{
  shapes.multipart,
  matrix,
  positioning,
  shapes.callouts,
  shapes.arrows,
  calc}

\usepackage{answers}
\usepackage{xcolor}
\usepackage{setspace}
\usepackage{graphicx}
\usepackage[shortlabels]{enumitem}
\usepackage{multicol}
\usepackage{mathrsfs}
\usepackage[margin=1in]{geometry} 
\usepackage{amsmath,amsthm,amssymb}
\usepackage{mathtools}
\usepackage[utf8]{inputenc}
\usepackage[english]{babel}
\usepackage{csquotes}
\usepackage[maxalphanames=99, maxbibnames=99,style=alphabetic,natbib=true]{biblatex}
\usepackage{bbm}
\usepackage{breqn}
\usepackage{subfiles}
\usepackage{amsmath,amssymb}
\usepackage[english]{babel}
\usepackage[nottoc]{tocbibind}
\usepackage{bm}
\usepackage{hyperref}
\usepackage{algorithm}
\usepackage{multirow}
\usepackage{multicol}
\usepackage{algpseudocode}

\usepackage{thm-restate}
\usepackage{tabularx}

\usepackage{gensymb}

\newtheorem{theorem}{Theorem}[section]

\newtheorem{claim}[theorem]{Claim}
\newtheorem{corollary}[theorem]{Corollary}
\newtheorem{lemma}[theorem]{Lemma}

\newtheorem{remark}[theorem]{Remark}

\newtheorem{proposition}[theorem]{Proposition}

\DeclarePairedDelimiter\floor{\lfloor}{\rfloor}

\newcommand{\R}{\mathbb{R}}

\newcommand{\Omegat}{\widetilde{\Omega}}
\DeclareMathOperator*{\E}{\mathbb{E}}
\DeclareMathOperator*{\argmax}{arg\,max}
\DeclareMathOperator*{\poly}{poly}
\DeclareMathOperator*{\supp}{supp}

\newcommand{\sign}{\text{sign}}

\newcommand{\eps}{\varepsilon}

\newcommand{\norm}[1]{\left\lVert #1\right\rVert}
\newcommand{\nnorm}[1]{\lVert #1\rVert}
\newcommand{\abs}[1]{|#1| }
\newcommand{\inner}[1]{\langle #1\rangle }
\newcommand{\wh}{\widehat}
\newcommand{\wt}{\widetilde}
\newcommand{\Ot}{\wt{O}}

\newcommand{\D}{\mathcal{D}}

\usepackage{xifthen}
\usepackage[capitalise]{cleveref}
\newcommand{\Prb}[2][]{ \ifthenelse{\isempty{#1}}
  {\Pr\left[#2\right]}
  {\Pr_{#1}\left[#2\right]} }
\newcommand{\Ex}[2][]{ \ifthenelse{\isempty{#1}}
  {\E\left[#2\right]}
  {\E_{#1}\left[#2\right]} }
\newcommand{\Vari}[2][]{ \ifthenelse{\isempty{#1}}
  {\mathbf{Var}\left[#2\right]}
  {\mathop{\mathbf{Var}}_{#1}\left[#2\right]} }
\newcommand{\nPrb}[2][]{ \ifthenelse{\isempty{#1}}
  {\Pr[#2]}
  {\Pr_{#1}[#2]} }
\newcommand{\nEx}[2][]{ \ifthenelse{\isempty{#1}}
  {\E[#2]}
  {\E_{#1}[#2\right]} }

\hypersetup{
    colorlinks=true,
    linkcolor=blue,
    filecolor=magenta,   
    urlcolor=cyan,
    citecolor=Green
}

\title{Spectral Guarantees for Adversarial Streaming PCA}
\author{Eric Price\\UT Austin \and Zhiyang Xun\\UT Austin}
\begin{document}
\begin{titlepage}
\maketitle
\begin{abstract}
  In streaming PCA, we see a stream of vectors
  $x_1, \dotsc, x_n \in \R^d$ and want to estimate the top eigenvector
  of their covariance matrix.  This is easier if the spectral ratio
  $R = \lambda_1 / \lambda_2$ is large.  We ask: how large does $R$
  need to be to solve streaming PCA in $\Ot(d)$ space?  Existing
  algorithms require $R = \Omegat(d)$.  We show:

  \begin{itemize}
  \item For all mergeable summaries, $R = \Omegat(\sqrt{d})$ is necessary.
  \item In the insertion-only model, a variant of Oja's algorithm
    gets $o(1)$ error for $R = O(\log n \log d)$.
      \item No algorithm {with $o(d^2)$ space} gets $o(1)$ error for $R = O(1)$.
  \end{itemize}

  Our analysis is the first application of Oja's algorithm to
  adversarial streams.  It is also the first algorithm for adversarial
  streaming PCA that is designed for a \emph{spectral}, rather than
  \emph{Frobenius}, bound on the tail; and the bound it needs is
  exponentially better than is possible by adapting a Frobenius
  guarantee.

\end{abstract}
  \thispagestyle{empty}
\end{titlepage}

\section{Introduction}

Principal Component Analysis (PCA) is a fundamental primitive for
handling high-dimensional data by finding the highest-variance
    directions.  At its most simple, given a data set
$X \in \R^{n \times d}$ of $n$ data points in $d$ dimensions, we want
to find the top unit eigenvector $v^*$ of the covariance matrix
$\Sigma = \frac{1}{n} X^T X$.

One common way to approximate $v^*$ is the power method: start with a
random vector $u_0$, then repeatedly multiply by $\Sigma$ and
renormalize.  This converges to $v^*$ at a rate that depends on the
ratio of the top two eigenvalues of $\Sigma$, denoted    $R := \lambda_1/\lambda_2$.  In
particular, after $O(\log_R \frac{d}{\eps})$ iterations we have
$\norm{Pu_k}^2 = 1 - \inner{u_k, v^*}^2 = \sin^2 (u_k, v^*) \leq \eps$ with high
probability, where $P = I - v^* (v^*)^T$ projects away from $v^*$.

But what if the data points $x_1, x_2, \dotsc, x_n \in \R^d$ arrive in
a streaming fashion?  Directly applying the power method requires
either $nd$ space to store $X$, or $d^2$ space to store $\Sigma$.
What can be done in smaller space?  The question of streaming PCA has
been extensively studied, in two main settings: \emph{adversarial} and
\emph{stochastic} streams.

In the adversarial streaming setting, we want to solve PCA for an
arbitrary set of data points in arbitrary order.  Many of these
algorithms store linear sketches of the data, such as $AX$ and $XB$
for Gaussian matrices
$A,B$~\cite{clarkson2009numerical,boutsidis2016optimal,woodruff2014sketching,upadhyay2016fast,tropp2017practical}.
These results give a Frobenius guarantee for rank-$k$ approximation of
$X$.  Specialized to $k=1$, the result direction $\wh{u}$ satisfies
\[
  \norm{X(I - \wh{u}\wh{u}^T)}_F^2 \leq (1 + \eps) \norm{XP}_F^2
\]
which is equivalent to
\[
  \wh{u}^T \Sigma \wh{u} \geq \lambda_1 - \eps \sum_{i > 1} \lambda_i.
\]
The best result here is
\textsc{FrequentDirections}~\cite{liberty2013simple,ghashami2016frequent},
which is a deterministic insertion-only algorithm rather than a linear
sketch.  It uses $O(d/\eps)$ space to get the guarantee, which is
optimal~\cite{clarkson2009numerical}.  Unfortunately, this Frobenius
guarantee can be quite weak: if the eigenvalues do not decay and we
only have a bound on $R = \lambda_1 / \lambda_2$, to get
$\norm{P\wh{u}}^2 \leq 0.1$ we need $\eps < \frac{R}{d}$, which means
$\Theta(d^2/R)$ space.  The well-known spiked covariance mode \cite{johnstone2001distribution}, where
the $x_i$ are iid Gaussian with covariance that has eigenvalues
$\lambda_2 = \lambda_3 = \dotsb = \lambda_d$, is one example where
this quadratic space bound appears.

In the stochastic streaming setting, the $x_i$ are drawn iid from a
somewhat nice distribution.  The goal is to converge to the principal
component of the true distribution using little space and few samples.
Algorithms for the stochastic setting are typically iterative, using
$O(d)$ space and converging to the true solution with a sample
complexity depending on how ``nice'' the distribution is.  Examples
include Oja's
algorithm~\cite{oja1982simplified,balsubramani2013fast,jain2016streaming,allen2017first,huang2021matrix,huang2021streaming,lunde2021bootstrapping}
and the block power
method~\cite{arora2012stochastic,mitliagkas2013memory,hardt2014noisy,balcan2016improved}.
Oja's algorithm starts with a random $v_0$, then repeatedly sets
\[
  v_i = v_{i-1} + \eta_i x_i x_i^T v_{i-1}
\]
for some small learning rate $\eta_i$.  These analyses depend
heavily on the data points being iid\footnote{Or nearly so; for
  example,~\cite{jain2016streaming} requires that the $x_i$ are independent with
  identical covariance matrices.}.  In return, they can get a stronger
\emph{spectral} guarantee than the sketching algorithms.  The bounds
are not directly comparable to the sketching algorithms (not only does
the sample complexity depend on the data distribution, but the
convergence is to the principal component of the true distribution
rather than the empirical $\Sigma$), but in the spiked covariance
setting they just need $n \geq \Ot((1 + \frac{1}{R-1})^2d)$ rather than
$O(d^2/R)$.  That is, they use near-linear samples down to
$R = 1 + \eps$.

So the situation is: algorithms that handle arbitrary data need
$O(d^2/R)$ space for a spectral guarantee.  Iterative methods have a
good spectral guarantee---linear space and often near-linear samples
for constant $R$---but only handle iid data.  Is this separation
necessary, or can we get a good spectral guarantee in the
arbitrary-data setting?  In this paper we ask:

\begin{quote}
  \emph{Is a polynomial spectral gap necessary to guarantee a near-linear space algorithm?}
\end{quote}

\subsection{Our results}
Our main result is that linear space \emph{is} possible for
polylogarithmic spectral gaps.  In fact, Oja's method essentially
achieves this:

\begin{restatable}[Performance of Oja's method in adversarial streams]{theorem}{thmOja}
  \label{thm:Oja1}
  For any sufficiently large universal constant $C$, suppose $\eta$ is
  such that $\eta n\lambda_1 > C\log d$ and
  $\eta n\lambda_2 < \frac{1}{C\log n}$.  If
  $\eta \norm{x_i}^2 \leq 1$ for every $i$, then Oja's algorithm with
  learning rate $\eta$ returns $\wh{v}$ satisfying
  $\norm{P\wh{v}} \leq \sqrt{\eta n\lambda_2} + d^{-9}$ with
  $1 - d^{-\Omega(C)}$ probability.

  Moreover, Oja's method can be modified (Algorithm~\ref{alg:oja}) so
  that in addition, regardless of $\lambda_1$ and $\lambda_2$, if
  $\eta \norm{x_i}^2 \leq 1$ for all $i$ then either
  $\norm{P\wh{v}} \leq \sqrt{\eta n\lambda_2} + d^{-9}$ or
  $\wh{v} = \perp$.
\end{restatable}

If $R > O(\log n \log d)$, there exists an $\eta$ that satisfies the
eigenvalue condition.  However, Theorem~\ref{thm:Oja1} requires
knowing $\eta$ and that no single $\norm{x_i}$ is too large.  It's
fairly easy to extend the algorithm to remove both restrictions, as
well as describe the performance with respect to finite precision.
Algorithm~\ref{alg:unknownrate} simply runs Oja's method for different
learning rates and picks the smallest one that works; unless any
single $x_i$ has too large $\norm{x_i}^2$ violating
Theorem~\ref{thm:Oja1}, in which case it outputs that $x_i$.
For $X \in \R^{n \times d}$ with $b$-bits entries, where each $X_{i, j}$ is either 0 or falls within $2^{-b} \leq \abs{X_{i,j}} \leq 2^{b}$, it suffices to test roughly $O(b)$ different learning rates in parallel. We say an algorithm $\eps$-approximates PCA if it returns $u$ with
$\norm{Pu}^2 \leq \eps$, and we have the following theorem.

\begin{restatable}[Full algorithm]{theorem}{thmupper}\label{thm:upper}
  For $X \in \R^{n \times d}$ have $b$-bit entries 
  for $b > \log(dn)$. 
  Whenever the spectral gap $R = \lambda_1 / \lambda_2 > O(\log n \log d)$,
  Algorithm~\ref{alg:unknownrate} uses $O(b^2 d)$ bits of space and
  $O(\frac{\log d}{R} + d^{-9})$-approximates PCA 
  with high
  probability.
\end{restatable}

\begin{algorithm}
  \caption{Oja's Algorithm, checking the growth of $\norm{v_n}$ to
    identify too-small learning rates.}\label{alg:oja}
\begin{algorithmic}
  \Procedure{OjaCheckingGrowth}{$X$, $\eta$}
    \State Choose $\wh{v}_0 \in S^{d-1}$ uniformly.
    \Comment{All numbers stored to $O(\log (nd))$ bits of precision}
    \State Set $s_0 = 0$.
    \For{$i = 1, \dotsc, n$}
      \State $v_i' \gets (1 + \eta x_i x_i^T)\wh{v}_{i-1}$.
      \State $\wh{v}_i \gets \frac{v'_i}{\norm{v'_i}}$.
      \State $s_i \gets s_{i-1} + \log \norm{v'_i}$.
    \EndFor
    \State \textbf{if} $s_n \leq 10\log d$, \Return $\perp$.
    \Comment{Returns $\perp$ rather than a wrong answer if $\eta$ is too small.}
    \State \textbf{else} \Return $\wh{v}_n$
  \EndProcedure
\end{algorithmic}
\end{algorithm}

\begin{algorithm}
\caption{Algorithm handling unknown learning rate and large-norm entries}\label{alg:unknownrate}
\begin{algorithmic}
  \Procedure{AdversarialPCA}{$X$, $b$}  \Comment{$X \in \R^{n \times d}$ has $X_{i,j} = 0$ or $2^{-b} \leq \abs{X_{i,j}} \leq 2^{b}$}
    \State Define $\eta_i = 2^i$ for integer $i$, $\abs{i} \leq 4 b + \log(nd^2) + O(1)$.
    \State In parallel run $\textsc{OjaCheckingGrowth}(X, \eta_i)$ for all $i$, getting $v^{(i)}$.
    \State In parallel record $\overline{x}$, the single $x_i$ of maximum $\norm{x_i}$.
    \State Let $i^*$ be the smallest $i$ with $v^{(i)} \neq \perp$.
    \State \textbf{if} $\eta_{i^*} \norm{\overline{x}}_2 \geq 1$, \Return $\frac{\overline{x}}{\norm{\overline{x}}}$.
    \State \textbf{else} \Return $v^{(i^*)}$.
  \EndProcedure
\end{algorithmic}
\end{algorithm}

\paragraph{Lower bound for mergeable summaries.}
Existing algorithms for adversarial PCA, including linear sketching or
\textsc{FrequentDirections}, fall under the category of mergeable
summaries~\cite{agarwal2013mergeable}. These algorithms enable
processing of disjoint data inputs on separate machines, producing
summaries that can be combined to address the problem using the full
dataset. By contrast, our algorithm is not mergeable and requires the
data to appear in one long sequence.

Considering the benefits of mergeable summaries, a natural goal would
be to get a good spectral guarantee with a mergeable summary.  As
discussed above, existing algorithms require $\Omega(d^2/R)$ space, so
$R = \wt{\Theta}(d)$ is needed for them to achieve near-linear space. 
Is it possible to get near-linear space and logarithmic $R$, like
Theorem~\ref{thm:upper} achieves in the insertion-only model?

Existing lower bounds~\cite{LiWoodruff} imply that $\Omega(d^2 / R^2)$ space is necessary for linear sketching (see \cref{app:lw} for discussion). 
We show that the same bound applies to \emph{all} mergeable algorithms: all mergeable summaries require
$\Omega(d^2/R^2)$  bits of space to $0.1$-approximate PCA, making
$R = \widetilde{\Omega}(\sqrt{d})$ necessary for near-linear space.

\begin{restatable}[Mergeable Lower Bound]{theorem}{mergeablelowerbound}
    \label{thm:sketching}
      For all mergeable summaries, $0.1$-approximate PCA on streams with spectral
  gap $R$ requires at least $\Omega({d^2}/{R^2})$ bits of space.
\end{restatable}

\paragraph{Dependence on Accuracy.}

Theorem~\ref{thm:upper} shows that it is possible to solve
$O(\frac{\log d}{R})$-approximate PCA in near-linear space.  This is
$o(1)$, but cannot be driven towards $0$ in the way that other
settings allow (in the iid setting, the accuracy improves
exponentially in the number of samples; in the existing
$O(d^2/R)$-space worst-case algorithms, the space grows as
$\frac{1}{\eps}$ for accuracy $\eps$).  Unfortunately, we show that
this is inherent: there is a phase transition where aiming for more
than $\poly(1/R)$ accuracy requires quadratic rather than near-linear
space.

\begin{restatable}[Accuracy Lower Bound]{theorem}{lowerbound}\label{thm:lower}
  There exists a universal constant $C > 1$ such that: for any
  $R > 1$, $\frac{1}{C R^2}$-approximate PCA on streams with spectral
  gap $R$ requires at least $\frac{d^2}{C R^3}$ bits of
  space for sufficiently large $d > \poly(R)$.
\end{restatable}

Specializing to constant $R$ gives the following corollary:

\begin{theorem}
  \label{thm:constant_R}
  For any constant $R > 1$,
  there exists a constant $\eps>0$ such that $\eps$-approximate PCA on streams of spectral gap $R$ requires $\Omega(d^2)$ bits of space.
\end{theorem}

This shows that for constant $R$, storing the entire covariance matrix is essentially the only thing one can do to achieve $o(1)$ accuracy.
By contrast, Theorem~\ref{thm:upper} shows
that for $R = \Theta(\log n \log d)$, $\eps$-solving PCA for any constant
$\eps > 0$ is possible in $\wt{O}(d)$ bits of space.
This is a much lower threshold than the $R = \wt{\Theta}(d)$
needed for near-linear space by existing analyses.

Our results are summarized in Table~\ref{table:results}, which gives
upper and lower bounds for the requirements for near-linear space.

\begin{table}[ht]
\centering
\label{table:samplers}

\makebox[\textwidth][c]{
\begin{tabular}{|c|c|c|c|c|}
\hline
Setting & Method    & Mergeable? & Requirement for $\wt{O}(d)$ space & Citation\\
  \hline
  \hline
  \multirow{2}{*}{Distributional} & Oja's algorithm & No & $\lambda_1 > \lambda_2$ & \cite{oja1982simplified}\\
   & Block power method & No & $\lambda_1 > \lambda_2$ & \cite{hardt2014noisy,balcan2016improved}\\
  \hline
  \hline
  \multirow{3}{*}{Adversarial} & Linear Sketching & Yes & $\lambda_1 > (\lambda_2 + \dotsc + \lambda_d) \cdot \Omega(\frac{1}{\log^C d})$& \cite{upadhyay2016fast,tropp2017practical}\\
   & \textsc{FrequentDirections} & Yes & $\lambda_1 > (\lambda_2 + \dotsc + \lambda_d)\cdot \Omega(\frac{1}{\log^C d})$& \cite{liberty2013simple,ghashami2016frequent}\\
        & Algorithm~\ref{alg:unknownrate} & No & $\lambda_1 > \lambda_2 \cdot O(\log d \log n)$& Theorem~\ref{thm:upper}\\
  \hline
  \hline
  
    \multirow{2}{*}{Adversarial}    & \multirow{2}{*}{Impossibility} & Yes & $\lambda_1 \le \lambda_2 \cdot d^{0.49}$& Theorem~\ref{thm:sketching}\\
        &  & No & $\lambda_1 \le \lambda_2 \cdot 100$& Theorem~\ref{thm:constant_R}\\
\hline
\end{tabular}
}
\caption{In various settings, the requirement on the eigenvectors
  $\lambda_i$ of the covariance matrix for the algorithm to
  get small constant approximate PCA in $\wt{O}(d)$ space.  In the last two rows,
  we instead state a setting of $\lambda_1, \lambda_2$ for an instance in
  which $\Ot(d)$ space is impossible.}\label{table:results}
\end{table}

\subsection{Related Work}

Oja's algorithm has been extensively studied in the stochastic
setting where the data streams are sampled iid; see,
e.g.,~\cite{balsubramani2013fast,jain2016streaming,allen2017first,huang2021matrix,huang2021streaming,lunde2021bootstrapping}.
Since the goal in this setting is to approximate the underlying
distribution's principal components, there is a minimum sample
complexity for even an offline algorithm to estimate the principal
component.  This line of work~\cite{jain2016streaming} can show that
Oja's algorithm has a similar sample complexity to the optimal offline
algorithm, even for spectral ratios $R$ close to $1$.
Recent work of \cite{kumar2024streaming} extends Oja's algorithm to data sampled from a Markov chain instead of iid samples. They showed that, despite the data dependency inherent in Markovian data, the performance of Oja's algorithm is as good as the iid case when the Markov chain has large second eigenvalue.

Our analysis of Oja's algorithm is by necessity quite different from
these stochastic-setting analyses.  Oja's algorithm returns
$v_n = B_n v_0$ for a transformation matrix
$B_n = (I + \eta_n x_n x_n^T)(I + \eta_{n-1} x_{n-1}
x_{n-1}^T)\dotsb(I + \eta_1 x_1 x_1^T)$.  In the stochastic setting,
$B_n$ is a random variable, with
$\E[B_i \mid B_{i-1}] = (I + \eta \Sigma) B_i$; the analyses focus on
matrix concentration of $B_n$, essentially to bound the deviation of
$B_n$ around the ``expected'' $(I + \eta \Sigma)^n$. 
In our arbitrary-data setting, $B_n$ is not a random variable at all.
The only randomness is the initialization $v_0$.  This makes our
analysis quite different, instead tracking how much $\wh{v}_i$ can
move under the covariance constraints.

Our lower bound construction for high accuracy is closely related to one
in~\cite{woodruff2014low}, which shows an $\Omega(dk/\eps)$ lower
bound for a $(1+\eps)$-approximate rank-$k$ approximation of $\Sigma$
in Frobenius norm.  The~\cite{woodruff2014low} construction for
$k = 1$ and $\eps = \Theta(\frac{1}{n})$ is very similar to ours, and
would give an $\Omega(d^2)$ lower bound for a small constant
approximation when $R < 2$. Our construction has a more careful analysis in terms of $R$.

Much of the prior work on streaming PCA, for both the adversarial and
stochastic settings, is focused on solving $k$-PCA not just the single
top direction.  We leave the extension of our upper bound to general
$k$ as an open question. 

\section{Proof Overview}

\subsection{Upper Bound}
For our application of Oja's algorithm we use a fixed learning rate $\eta$
throughout the stream.
The $x_i$ correlated with $v^*$ could all
arrive at the beginning or the end of the stream, and we want to
weight them equally so that at least we can solve the commutative case
where Oja's algorithm is relatively simple.

As a basic intuition, Oja's algorithm
returns $\wh{v}_n = \frac{v_n}{\norm{v_n}}$, where
\begin{align*}
  v_n &= (I + \eta x_n x_n^T)(I + \eta x_{n-1} x_{n-1}^T)\dotsb(I + \eta x_1 x_1^T) v_0\\
      &\approx e^{\eta x_n x_n^T}e^{\eta x_{n-1} x_{n-1}^T}\dotsb e^{\eta x_{1} x_{1}^T}v_0
\end{align*}
where the approximation is good when $\eta \norm{x_i}^2 \ll 1$.
Imagine that these matrix exponentials commute (e.g., each $x_i$ is
$e_j$ for some $j$).  Then we would have
\begin{align}
  \label{eq:ojaapprox}
  v_n \approx e^{\eta X^T X} v_0.
\end{align}
This suggests that the important property of the learning rate $\eta$
is the spectrum of $\eta X^T X$.  Let $\eta X^T X$ have top eigenvalue
$\sigma_1 = n \eta \lambda_1$, with corresponding eigenvector $v^*$,
and all other eigenvalues at most $\sigma_2 = n \eta \lambda_2$.  For
Theorem~\ref{thm:Oja1}, we would like to show that Oja's algorithm
works if $\sigma_1 > O(\log d)$ and $\sigma_2 < \frac{1}{O(\log n)}$.

For~\eqref{eq:ojaapprox} to converge to $v^*$, as in the power method,
we want the $v^*$ coefficient of $v_0$ to grow by a $\poly(d)$ factor
more than any other eigenvalue, i.e.,
$e^{\sigma_1} \geq \poly(d) e^{\sigma_2}$ or
$\sigma_1 \geq \sigma_2 + O(\log d)$.  So we certainly need to set
$\eta$ such that $\sigma_1 \geq O(\log d)$.  But how large a spectral
gap do we need, i.e., how small does $\sigma_2$ need to be?

One big concern for adversarial-order Oja's algorithm is: even if most
of the stream clearly emphasizes $v^*$ so $v_i$ converges to it, a
small number of inputs at the end could cause $v_n$ to veer away from
$v^*$ to a completely wrong direction.  This cannot happen in the
commutative setting, but it can happen in general: $v_n$ can rotate by
$\Theta(\sqrt{\sigma_2})$, by ending the stream with $\frac{1}{\eta}$
copies of $v^* + \sqrt{\sigma_2} v'$ (see Figure~\ref{fig:endrotate}).
But this is the worst that can happen.  We show:

\begin{figure}
  \begin{center}
    \begin{tikzpicture}
      \coordinate (O) at (0,0);
      \coordinate (xcoord) at (30:2);
      \coordinate (vcoord) at (0:2);
      \coordinate (newcoord) at ($ (vcoord) + (30:1.73) $);
      \draw[thick] (O) circle (2);
      \draw[red, thick,->] (O) -- (vcoord);
      \draw[blue, dotted, thick,->] (O) -- (newcoord);
      \node[anchor=south,blue,shift={(7mm,0)}] (xn) at (newcoord) {$\wh{v}_{n-1} + \eta x_n x_n^T \wh{v}_{n-1}$};
      \node[red,anchor=west] (v) at (vcoord) {$\wh{v}_{n - 1} = v^*$};
      \draw[thick,->] (0, 0) -- (xcoord);
      \node[anchor=south west] (x) at (xcoord) {$x_n$};
      \draw pic [draw,->,angle radius=1cm,"$\sqrt{\sigma_2}$" shift={(7mm,3mm)}] {angle = vcoord--O--xcoord};
  \end{tikzpicture}
  \end{center}
  \caption{Suppose $\eta = 1$.  Then even after convergence to $v^*$
    exactly, a single final sample can skew the result by
    $\Theta(\sqrt{\sigma_2})$.  For smaller $\eta$, the same can
    happen with $\frac{1}{\eta}$ final samples. }\label{fig:endrotate}
\end{figure}
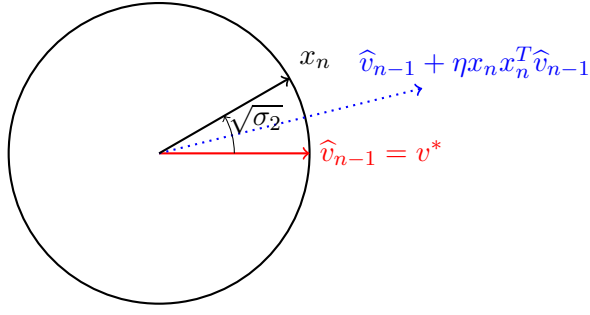

\begin{restatable}[Growth implies correctness]{lemma}{growthgivescorrectness}
  \label{lem:growth-gives-correctness}
  For any $v_0$ and all $i$, $\norm{P\wh{v}_i}  \leq \sqrt{\sigma_2} + \frac{\norm{Pv_0}}{\norm{v_i}}$.
\end{restatable}

This lemma has two useful implications: first, if we ever get close to
$v^*$, the final solution will be at most $\sqrt{\sigma_2}$ further from
$v^*$.  Second, no matter where we start, the final output is good if
$\norm{v_n}$ is very large.  This is how Algorithm~\ref{alg:oja} can
return either a correct answer or $\perp$: it just observes whether
$\norm{v_n}$ has grown by $\text{poly}(d)$.

So it suffices to show that $\norm{v_n}$ is large for a random $v_0$;
and since $v_0$ starts with a random $\frac{1}{\poly(d)}$ component in
the $v^*$ direction, it in fact suffices to show that $\norm{v_n}$
would grow by $\poly(d)$ if Oja's algorithm started at $v_0 = v^*$.
Now, one can show that
\begin{align}
\norm{v_n}^2 \geq  e^{\sum_{i=1}^n \eta \inner{x_i, \wh{v}_{i-1}}^2}.\label{eq:5}
\end{align}
So if $v_i$ were always exactly $v^*$, we would have
$\norm{v_n}^2 \geq e^{\eta (v^*)^T X^T X v^*} = e^{\sigma_1} \geq \poly(d)$
as needed. In addition, if we start at $v^*$, then
Lemma~\ref{lem:growth-gives-correctness} implies
$\norm{P \wh{v}_i} \leq \sqrt{\sigma_2}$ for all $i$, so we never
deviate \emph{much} from $v^*$. However, $v_i$ can deviate a little bit,
which could decrease $\inner{x_i, \wh{v}_{i-1}}^2$.  The question is, by how much?
Well, it's easy to show
\begin{align}
  \eta \inner{x_i, \wh{v}_{i-1}}^2 \geq \eta \frac{1-\sigma_2}{2}\inner{x_i, v^*}^2 - \eta \inner{x_i, P \wh{v}_{i-1}}^2\label{eq:4}
\end{align}
so we just need to show that
\begin{align}
  \eta \sum_i \inner{x_i, P \wh{v}_{i-1}}^2 \ll \sigma_1.\label{eq:xPvbound}
\end{align}
We know that $\norm{P \wh{v}_{i-1}}^2 \leq \sigma_2$, and
$\eta \sum_i \inner{x_i, w}^2 \leq \sigma_2$ for any fixed unit vector
$w \perp v^*$, but the worry is that $P\wh{v}_{i-1}$ could rotate
through many different orthogonal directions; each direction $w$ can
only contribute $\sigma_2^2$ to
$\eta \sum_i \inner{x_i, P \wh{v}_{i-1}}^2$, but the total could
conceivably be up to $\sigma_2^2d$.

Our main technical challenge is to rule this out, so
$\eta \sum_i \inner{x_i, P \wh{v}_{i-1}}^2$ is small.  For intuition,
in this overview we just rule out $P \wh{v}_{i-1}$ moving through many
\emph{standard} basis vectors by showing
\begin{align}
  \sum_{j=1}^d \max_i \inner{e_j, P \wh{v}_{i-1}}^2 \lesssim \sigma_2 \log^2 n \log \norm{v_n}.\label{eq:eiprod}
\end{align}
That is, $P \wh{v}_{i-1}$ cannot rotate through $\sqrt{\sigma_2}$
correlation with each of the $d$ different basis vectors (which would
give a value of $\sigma_2 d$) unless $\norm{v_n}$ is large (which is
what we wanted to show in the first place).

First, we show that $\norm{v_n}$ grows proportional to the
\emph{squared} movement of $P\wh{v}_i$:

\begin{restatable}{lemma}{normalizedmovement}\label{lem:normalized-movement}
  Suppose $Pv_0 = 0$.  For any two time steps $0 \leq a < b \leq n$,
  \[
    \norm{P\wh{v}_b - P\wh{v}_a}^2 
   \leq 4 \sigma_2 \log \frac{\norm{v_b}}{\norm{v_a}}
  \]
\end{restatable}
As a result, for any subsequence $i_0, \dotsc, i_k$ of iterations, the sum
of squared movement has
\[
  \sum_{j=1}^k \norm{P\wh{v}_{i_j} - P\wh{v}_{i_{j-1}}}^2 \lesssim \sigma_2 \log \norm{v_n}.
\]
We use a combinatorial lemma to turn this bound on squared distances
over subsequences into~\eqref{eq:eiprod}.  For any set of vectors $A$
the following holds (see Figure~\ref{fig:combinatorial}):
\begin{lemma}[Simplified version of Lemma~\ref{lem:matsample}]\label{lem:matsamplesimple}
  Let $A_0 = 0$, and $A_1, \dotsc, A_n \in \R^d$ satisfy that every
  subsequence $S$ of $\{0, \dotsc, n\}$ has
  \[
    \sum_i \norm{A_{S_i} - A_{S_{i-1}}}_2^2 \leq B.
  \]
  for some $B > 0$.  Then
  \[
    \sum_{j=1}^d \max_{i \in [n]} (A_i)_j^2 \leq B (1 + \log_2 n)^2.
  \]
\end{lemma}
Applying Lemma~\ref{lem:matsamplesimple} to $A_i := P \wh{v}_i$
immediately gives~\eqref{eq:eiprod}.

\begin{remark}
  The $\log^2 n$ factor in Lemma~\ref{lem:matsamplesimple} is why we
  need $R > O(\log d \log n)$, rather than just $R > O(\log d)$.  The
  factor in Lemma~\ref{lem:matsamplesimple} is tight for
  $n = \Theta(d)$: $A_{i,j} := \log \frac{n}{1 + \abs{i - j}}$ has
  $B = \Theta(n)$ while $\sum_{j=1}^d \max_{i \in [n]} (A_i)_j^2$ is
  $\Theta(n \log^2
  n)$.
\end{remark}

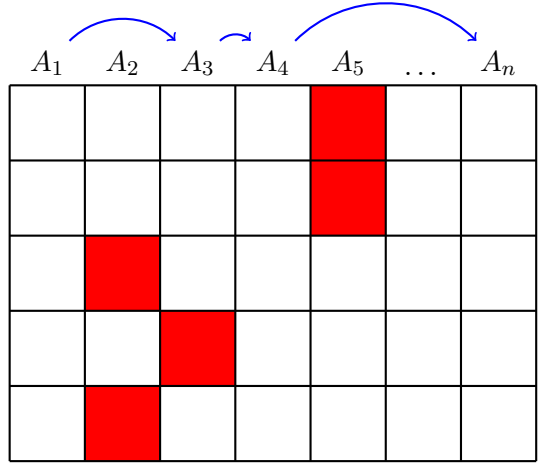
\begin{figure}
  \begin{center}
    \begin{tikzpicture}[
      myrect/.style={
        draw,thick,
        fill=white,
      }
      ]
      \newcommand{\cellsize}{1cm}
      \newcommand{\Rows}{5}
      \newcommand{\Cols}{7}
      \foreach \r / \c in {0/1, 1/2, 2/1, 3/4, 4/4} {
        \draw [myrect,draw=none,fill=red] ($ (\c*\cellsize, \r*\cellsize) $) rectangle ($ (\c*\cellsize+\cellsize, \r*\cellsize+\cellsize) $);
      }
      \foreach \r  in {0, 1,..., \Rows} {
        \draw[thick]  ($ (0, \r * \cellsize) $) --  ($ (\Cols * \cellsize, \r*\cellsize ) $);
      }
      \foreach \c  in {0,1,...,\Cols} {
        \draw[thick] ($ (\c*\cellsize, 0) $) --  ($ (\c *\cellsize, \Rows * \cellsize) $);
      }
      \foreach \c / \t  in {1/1,2/2,3/3,4/4,5/5,7/n} {
        \node[anchor=south] (A\c) at ($ (\c *\cellsize - .5 *\cellsize, \Rows * \cellsize) $) {$A_\t$};
      }
      \node[anchor=south] (A6) at ($ (6 *\cellsize - .5 *\cellsize, \Rows * \cellsize) $) {$\dotsc$};
      \draw[->,thick,blue] (A1) edge[bend left=45] (A3);
      \draw[->,thick,blue] (A3) edge[bend left=45] (A4);
      \draw[->,thick,blue] (A4) edge[bend left=45] (A7);
    \end{tikzpicture}
  \end{center}
  \caption{Lemma~\ref{lem:matsamplesimple} states that, if the sum of
    squared distances across any subsequence of vectors $A_i$ is at most $B$,
    then the vector selecting the maximum value in each coordinate has squared norm
    $B \log^2 n$.  }\label{fig:combinatorial}
\end{figure}

A similar approach, applied to $A_{i,j} = x_i^T P \wh{v}_j$, lets us
bound our actual target~\eqref{eq:xPvbound}:

\begin{restatable}{lemma}{uprodlem}\label{lem:uprod}
  If $v_0 = v^*$, then
  \[
    \eta \sum_{i=1}^n \inner{x_i, P \wh{v}_{i-1}}^2 \lesssim \sigma_2^2 \log^2 n \log \norm{v_n}
  \]
\end{restatable}
Combined with~\eqref{eq:5} and~\eqref{eq:4}, this implies that
$\norm{v_n} \geq e^{\Omega(\sigma_1)}$ if $\sigma_2 \ll \frac{1}{\log n}$:

\begin{restatable}[The right direction grows]{lemma}{rightgrowthlem}\label{lem:growth}
  Suppose $\sigma_2 < \frac{1}{2}$.  Then if $v_0 = v^*$ we have
  \[
    \log \norm{v_n} \gtrsim \frac{\sigma_1}{1 + \sigma_2^2 \log^2 n}.
  \]
\end{restatable}

Since the initial $v_0$ is random, it with high probability has a
$\frac{1}{\poly(d)}$ component in the $v^*$ direction; then linearity
of the unnormalized algorithm means $\norm{v_n}$ is large with high
probability.  By Lemma~\ref{lem:growth-gives-correctness}, this means
the angle between $v^*$ and final answer $\wh{v}_n$ is bounded by $\sqrt{\sigma_2} + d^{-C}$,  so the algorithm succeeds.

\subsection{Lower bound for mergeable summaries}
In this section, we outline the proof of the mergeable summary lower bound. Specifically, we show with spectral gap $R = o(\sqrt{d})$, no mergeable summary algorithm can 0.1-approximate PCA using $O(d)$ space. Consider the scenario where there are $R$ players each possessing $d / R - 1$ vectors that are i.i.d.~drawn from $\mathcal{N}(0, I_d)$. We then insert another vector $v^* \sim \mathcal{N}(0, I_d)$ into a random location in each player's data list. 
Consequently, from the viewpoint of each participant, their dataset consists of $d / R$ i.i.d.\ Gaussian vectors, making it impossible for them to individually identify the shared vector $v^*$.

Now we consider the spectral properties of the overall data. With high probability, the variance in the direction of $v^*$ will be $\Theta(R)$ times larger than orthogonal directions. This implies that
\begin{enumerate}
    \item The spectral gap of the data is at least $\Theta(R)$. 
    \item The principal component of the data is very close to $v^*$.
\end{enumerate}
Therefore, we can employ a mergeable summary algorithm that approximates PCA to approximate $v^*$. Now we only need to prove that each player's summary must have at least $\omega(d)$ bits.

Suppose each player runs this algorithm and writes an $s$ bits long
summary $S$ of their data to help approximate $v^*$.  We measure the
amount of information $S$ contains about $v^*$, i.e., $I(v^*; S)$.
Because each player cannot distinguish $v^*$ from the other $d/R - 1$
vectors they have, we can prove that
$I(S; v^*) \leq H(S) / (d/R) = R \cdot O(s / d)$. With $R = o(\sqrt{d})$
players, the combined summaries have at most
$R \cdot I(S; v^*) \le R^2 \cdot O(s / d) = o(s)$ bits of information about
$v^*$. Since an approximation of $v^*$ has $\Theta(d)$ bits of
information about $v^*$, this requires that $s = \omega(d)$.

\subsection{Lower bound for high accuracy in insertion-only streams}

To give an $\Omega(d^2)$ lower bound for constant $R$, we construct a
two-player one-way communication game, where Alice feeds a uniformly
random stream into the algorithm and passes the state to Bob.  Bob
then repeatedly takes this state, adds a few more vectors, and extracts
the PCA estimate.  We will show that Bob is able to learn
$\Omega(d^2)$ bits about Alice's input, and therefore the stream state
must have $\Omega(d^2)$ bits.  Our approach is illustrated in
Figure~\ref{fig:lower}.

\begin{figure}[b]
  \begin{center}
    \begin{tabular}{|cccc|cccc|}
      \hline
      1&-1&-1&1&-1& -1&1&1\\
      \textcolor{red}{\textbf{1}}&\textcolor{red}{\textbf{1}}&\textcolor{red}{\textbf{1}}&\textcolor{red}{\textbf{-1}}&\textcolor{blue}{\textbf{1}}& \textcolor{blue}{\textbf{1}}&\textcolor{blue}{\textbf{-1}}&\textcolor{blue}{\textbf{-1}}\\
      -1&1&1&-1&-1& -1&-1&-1\\
      -1&-1&1&-1&-1& -1&1&-1\\
      \hline
      \textcolor{red}{\textbf{1}}&\textcolor{red}{\textbf{1}}&\textcolor{red}{\textbf{1}}&\textcolor{red}{\textbf{-1}}& 0 & 0 & 0 & 0\\
      \textcolor{red}{\textbf{1}}&\textcolor{red}{\textbf{1}}&\textcolor{red}{\textbf{1}}&\textcolor{red}{\textbf{-1}}& 0 & 0 & 0 & 0\\
      \textcolor{red}{\textbf{1}}&\textcolor{red}{\textbf{1}}&\textcolor{red}{\textbf{1}}&\textcolor{red}{\textbf{-1}}& 0& 0 & 0 & 0\\
      \hline
    \end{tabular}
  \end{center}
  \caption{High-accuracy lower bound approach: Alice inserts a sequence of random bits (all but the last row).  Bob knows the left side and wants to approximate the right side.  To estimate the \textbf{\textcolor{blue}{blue}} bits on the right, he adds $O(1)$ vectors using the corresponding \textbf{\textcolor{red}{red}} bits on the left and random bits on the right.  With high probability, the principal component has constant correlation with the blue bits.}\label{fig:lower}
\end{figure}

Suppose that Alice feeds in a random binary stream
$x_1, x_2, \dotsc, x_n \in \{-1, 1\}^{d}$.  What can Bob insert so the
PCA solution reveals information about (say) $x_1$?

First, suppose Bob inserted $k-1$ more copies of $x_1$ for some
constant $k$.  Then (if $n < d/100$) the PCA solution would be very
close to $x_1$: $v = \frac{x_1}{\norm{x_1}}$ has
$\norm{Xv}^2 \geq k d$ from just the copies of $x_1$, while every
orthogonal direction has variance at most
$(\sqrt{n} + \sqrt{d})^2 \approx 1.1d$ by standard bounds on singular
values of subgaussian matrices~\cite{rudelson2010non}.  Thus the
spectral ratio $R = \frac{\lambda_1}{\lambda_2} > \frac{k}{1.1}$, so
the streaming algorithm should return a vector highly correlated with
$x_1$.  The problem with this approach is that Bob cannot insert $x_1$
without knowing $x_1$, so the streaming PCA solution does not reveal
any \emph{new} information to him.

But what if Bob inserts $z_2, \dotsc, z_k$ that match $x_1$ on the
first half of bits, and are all 0 on the second half?  The top
principal component $u^*$ will still be highly correlated with $x_1$:
the vector $v$ that matches $x_1, z_2, \dotsc, z_k$ on the first half
of bits and is zero on the rest has variance that is a
$O(k)$ factor larger than any orthogonal direction. 

A more careful analysis shows that the top principal component $v^*$ is
not only correlated with the half of bits of $x_1$ shared
with the $z_i$, but (on the remaining bits) is very highly correlated
with the average $\frac{1}{k}(x_1 + z_2 + \dotsb + z_k)$.  In fact, it
is \emph{so} highly correlated with the average that $v^*$ must be at
least somewhat---$\Theta(1/k^2)$---correlated with the last $10\%$ of
bits in $x_1$.  This analysis is robust to a PCA approximation, so the
streaming PCA algorithm lets Bob construct $\wh{v}$ with constant
correlation with the last half of bits in $x_1$.

Thus Bob can learn $\Omega(d)$ bits about the first row by inserting
$z_2, \dotsc, z_k$ that match the first half of bits and looking at the
PCA solution on the last half of bits.  If he does this for every
row, he learns $\Omega(nd) = \Omega(d^2)$ bits about Alice's input.
Therefore the algorithm state Alice sent needs $\Omega(d^2)$ space.

This construction is very similar to the one in~\cite{woodruff2014low}
for lower-bounding low-rank Frobenius approximation.  The difference
in~\cite{woodruff2014low} is that Bob only inserts one row, so
necessarily $R < 2$. Our main contribution here is the more careful
analysis in terms of $R$.

\section{Proof of Upper Bound}

For most of this section we focus on Oja's method
(Theorem~\ref{thm:Oja1}), then in Section~\ref{sec:upper} we show
Theorem~\ref{thm:upper}.
For simplicity, the proof is given assuming exact arithmetic.  In
Section~\ref{sec:precision} we discuss why $O(\log (nd))$ bits of
precision suffice.

\paragraph{Setup.}  $\wh{v}_i$ is the normalized state at time $i$,
$v_i$ is the unnormalized state, $x_i$ is the sample, $\eta$ is the
learning rate, $v^*$ is the direction of maximum variance,
$P = I - v^* (v^*)^T$ to be the projection matrix that removes the
$v^*$ component.  Let $\sigma_1 = \eta \norm{X^TX}$ and
$\sigma_2 = \eta \norm{P X^TXP}$, so:
\begin{align}
  \sum_{i=1}^n \inner{v^*, x_i}^2 &= \sigma_1\label{eq:bigdir}\\
  \eta \sum_{i=1}^n \inner{w, x_i}^2 &\leq \sigma_2 & (\forall w \perp v^*)\label{eq:otherdir}
\end{align}
For much of the proof we will also need $\sigma_1 \geq C\log d$
and $\sigma_2 \leq \frac{1}{C \log n}$, but this will be stated as
needed.

Oja's algorithm works by starting with a (typically random) vector
$v_0$, then repeatedly applying Hebb's update rule that ``neurons that
fire together, wire together'':
\begin{align}\label{eq:unnormalized-update}
  v_i = v_{i-1} + \eta \inner{x_i, v_{i-1}} x_i = (I + \eta x_i x_i^T)v_{i-1}.
\end{align}
The algorithm only keeps track of the normalized vectors
$\wh{v}_i = v_i / \norm{v_i}$, but for analysis purposes we will often
consider the unnormalized vectors $v_i$.

The norm $\norm{v_i}$ grows in each step, according to
\begin{align}
  \norm{v_{i}}^2 = \norm{v_{i-1}}^2 (1 + (2\eta + \eta^2\norm{x}^2) \inner{x_i, \wh{v}_{i-1}}^2),
\end{align}
and in particular (since Theorem~\ref{thm:Oja1} assumes $\eta \norm{x_i}^2 \leq 1$)
\begin{align}
  \log \frac{\norm{v_i}^2}{\norm{v_{i-1}}^2} \geq \eta \inner{x_i, \wh{v}_{i-1}}^2.\label{eq:normstep}
\end{align}
Our goal is to show that $\wh{v}_n \approx v^*$, or equivalently, that
$\norm{P \wh{v}_n}$ is small.

\subsection{Initial Lemmas}

\begin{claim}\label{claim:prodab}
  Let $0 \leq a_1, a_2, \dotsc, a_n$ and define $b_i = e^{\sum_{j \leq i} a_i}$ for $i \in \{0, 1, \dotsc, n\}$.  Then:
  \[
    \sum_{i=1}^n a_i b_{i-1} \leq b_n - 1.
  \]
\end{claim}
\begin{proof}
  This follows from induction on $n$.  $n=0$ is trivial, and then
  \[
    \sum_{i=1}^n a_i b_{i-1} \leq b_{n-1} - 1 + a_n b_{n-1} = (1 + a_n) b_{n-1} - 1 \leq e^{a_n} b_{n-1} - 1 = b_n - 1.
  \]
\end{proof}

Define $B_i = \frac{\norm{v_i}^2}{\norm{v_0}^2}$, and $A_i = \log \frac{B_i}{B_{i-1}}$
which satisfies $A_i \geq \eta \inner{x_i, \wh{v}_{i-1}}^2$
by~\eqref{eq:normstep}.  Therefore
  \begin{align}
    \eta \sum_{i=1}^n \inner{x_i, v_{i-1}}^2 \leq \norm{v_0}^2 \sum_{i = 1}^n A_i B_{i-1} \leq \norm{v_0}^2 (B_n-1) = \norm{v_n}^2 - \norm{v_0}^2
    \label{eq:normthing}
  \end{align}
  by Claim~\ref{claim:prodab}.  Then for any unit vector $w$ with $Pw = w$,
  \begin{align*}
    \inner{v_n - v_0, w}^2 &= \left(\eta \sum_{i=1}^n \inner{x_i, v_{i-1}} \inner{x_i, w}\right)^2 && \text{by~\eqref{eq:unnormalized-update}}\\
                         &\leq \eta \sum_{i=1}^n \inner{x_i, v_{i-1}}^2  \cdot \eta \sum_{i=1}^n\inner{x_i, w}^2&& \text{by Cauchy-Schwarz}\\
                         &\leq (\norm{v_n}^2 - \norm{v_0}^2) \sigma_2. &&\text{by~\eqref{eq:normthing} and \eqref{eq:otherdir}}
  \end{align*}
  There's nothing special about the start and final indices, giving the following bound for general indices
  $a \leq b$:
  \begin{align}
    \inner{v_b - v_a, w}^2 &\leq (\norm{v_b}^2 - \norm{v_a}^2) \sigma_2\label{eq:vwbound}.
  \end{align}

\growthgivescorrectness*

\begin{proof}
  By~\eqref{eq:vwbound}, for any $w$ with $w = Pw$,
  \[
    \inner{v_i - v_0, w} \leq \sqrt{\sigma_2} \norm{v_i}.
  \]
  Hence
  \begin{align*}
    \inner{\wh{v}_i, w} = \frac{\inner{v_i - v_0, w} + \inner{v_0, w}}{\norm{v_i}} \leq \sqrt{\sigma_2} + \frac{\inner{v_0, w}}{\norm{v_i}}.
  \end{align*}
  Setting $w = P\wh{v}_i / \norm{P \wh{v}_i}$, we have
  $\inner{\wh{v}_i, w} = \norm{P \wh{v}_i}$ and
  $\inner{v_0, w} \leq \norm{P v_0}$, giving the result.
\end{proof}

Lemma~\ref{lem:growth-gives-correctness} implies that, if we start at
$v^*$, we never move by more than $\sqrt{\sigma_2}$ from it.  We now
show that you cannot even move $\sqrt{\sigma_2}$ without increasing the
norm of $v$.

\normalizedmovement*
\begin{proof}
  Define $w$ to be the unit vector in direction $P(\wh{v}_b - \wh{v}_a)$.  By~\eqref{eq:vwbound} we have
  \[
    \inner{v_b - v_a, w}^2 \leq \sigma_2(\norm{v_b}^2 - \norm{v_a}^2).
  \]
  Therefore
  \begin{align*}
    \norm{P\wh{v}_b - P\wh{v}_a}^2 = \inner{P(\wh{v}_b - \wh{v}_a), w}^2 &= \inner{\wh{v}_b - \wh{v}_a, w}^2\\
    &\leq 2\inner{\wh{v}_b - \frac{\norm{v_a}}{\norm{v_b}} \wh{v}_a, w}^2 + 2\inner{\frac{\norm{v_a}}{\norm{v_b}} \wh{v}_a - \wh{v}_a, w}^2\\
    &\leq 2 \frac{1}{\norm{v_b}^2}\inner{v_b - v_a, w}^2 + 2 (\frac{\norm{v_a}}{\norm{v_b}} - 1)^2 \norm{P \wh{v}_a}^2\\
                                                                          &\leq 2\sigma_2 (1 - \frac{\norm{v_a}^2}{\norm{v_b}^2}) + 2(1 - \frac{\norm{v_a}}{\norm{v_b}})^2\sigma_2 \\
                                                                          &= 4\sigma_2 (1 - \frac{\norm{v_a}}{\norm{v_b}}).
  \end{align*}
  Finally, $(1 - 1/x) \leq \log x$ for all $x > 0$.
\end{proof}

\subsection{Results on Sequences}
The following combinatorial result is written in a self-contained
fashion, independent of the streaming PCA application.

\begin{lemma}\label{lem:matsample}
  Let $A \in \R^{d \times n}$ have first column all zero.  Define
  $b^{(k)}_i$ to be column $1 + 2^k i$ of $A$.  Then:
  \[
    \sum_i \max_j A_{ij}^2 \leq (1+\log_2 n) \sum_{k=0}^{\log_2 n} \sum_{j > 0} \norm{b^{(k)}_j - b^{(k)}_{j-1}}^2.
  \]
\end{lemma}
\begin{proof}
  We will show this separately for each row $i$; the result is just
  the sum over these rows.  For fixed $i$, let $j^* = \argmax_j A_{ij}^2$.

  Let $j^{(k)} = 1+2^k\floor{\frac{j^*-1}{2^k}}$ set the last $k$ bits of
  $j^*-1$ to zero.  We have that $j^{(0)} = j^*$ and $j^{\log_2 n} = 0$.  Therefore
  \[
    A_{ij^*} = \sum_{k=0}^{\log_2 n} (A_{i, j^{(k)}} - A_{i, j^{(k+1)}}).
  \]
  Now, $j^{(k)}$ is either $j^{(k+1)}$ or $j^{(k+1)} + 2^k$.  Each
  value in the right sum is either zero (if $j^{(k)}$ is $j^{(k+1)}$)
  or the $i$th coordinate of $b^{(k)}_{j'} - b^{(k)}_{j'-1}$ for some $j'$ (if
  $j^{(k)} = j^{(k+1)}+2^k$, using $j' = j^{(k)}/2^k$).  Thus, by
  Cauchy-Schwarz,
  \begin{align*}
    A_{ij^*}^2 &\leq (1+\log_2 n) \cdot \sum_{k=0}^{\log_2 n} (A_{i, j^{(k)}} - A_{i, j^{(k+1)}})^2\\
    &\leq (1+\log_2 n) \cdot \sum_{k=0}^{\log_2 n} \sum_{j > 0} ((b^{(k)}_j)_i - (b^{(k)}_{j-1})_i)^2.
  \end{align*}
  Summing over $i$,
  \[
    \sum_i \max_j A_{ij}^2 \leq (1+\log_2 n) \sum_{k=0}^{\log_2 n} \sum_{j > 0} \norm{b^{(k)}_j - b^{(k)}_{j-1}}^2.
  \]
\end{proof}

\subsection{Proof of Growth}

We return to the streaming PCA setting.  The goal of this section is
to show that, if $v_0 = v^*$, then $\norm{v_n}$ is large.

\uprodlem*
\begin{proof}
  Define $u_i = P \wh{v}_i$.
  We apply Lemma~\ref{lem:matsample} to the matrix $A_{ij} = \inner{x_i, u_{j-1}}$ for $i, j \in [n]$, getting:
  \[
    \sum_{i=1}^n \max_{j \leq n-1} \inner{x_i, u_j}^2 \leq (1+\log_2 n) \sum_{k=0}^{\log_2 n} \sum_{j > 0} \sum_{i=1}^n (\inner{x_i, u_{2^k j}} - \inner{x_i, u_{2^k (j-1)}})^2.
  \]
  Now,
  \begin{align*}
    \sum_{i=1}^n (\inner{x_i, u_{2^k j}} - \inner{x_i, u_{2^k (j-1)}})^2 &= (u_{2^k j} - u_{2^k (j-1)}) X^TX (u_{2^k j} - u_{2^k (j-1)})\\
    &\leq \frac{\sigma_2}{\eta} \norm{u_{2^k j} - u_{2^k (j-1)}}^2.
  \end{align*}
  by the assumption~\eqref{eq:otherdir} on $X$ and that every
  $u_j \perp v^*$.  Then, for each $k$, Lemma~\ref{lem:normalized-movement} shows that
  \[
    \sum_{j > 0} \norm{u_{2^k j} - u_{2^k (j-1)}}^2 \leq    4\sigma_2 \log \frac{\norm{v_n}}{\norm{v_0}} = 4\sigma_2 \log \norm{v_n}
  \]
  and thus
  \[
    \eta\sum_i \inner{x_i, P \wh{v}_{i-1}}^2 \leq \eta \sum_i \max_j \inner{x_i, u_j}^2 \leq (1 + \log_2 n) \sum_{k=0}^{\log_2 n} 4\sigma_2^2 \log \norm{v_n} \lesssim \sigma_2^2  \log^2 n \log \norm{v_n}
  \]
  as desired.
\end{proof}

\rightgrowthlem*
\begin{proof}
  We will show that $\eta \sum_{i=1}^n \inner{x_i, \wh{v}_{i-1}}^2 \gtrsim \sigma_1$, giving the result by~\eqref{eq:normstep}.

  Recall that $(x + y)^2 \geq \frac{1}{2}x^2 - y^2$ for all $x, y$.
  Thus, if $\wh{v}_i = a_i v^* + u_i$ for $u_i \perp v^*$, we have
  \[
    \inner{x_i, \wh{v}_{i-1}}^2 \geq \frac{a_{i-1}^2}{2}\inner{x_i, v^*}^2 - \inner{x_i, u_{i-1}}^2.
  \]
  Lemma~\ref{lem:growth-gives-correctness} shows that
  $a_i^2 \geq 1 - \sigma_2 \geq \frac{1}{2}
  $, so summing up over $i$,
  \[
    \eta \sum_{i=1}^n \inner{x_i, \wh{v}_{i-1}}^2 \geq \frac{1}{4} \sigma_1 - \eta \sum_{i=1}^n \inner{x_i, u_{i-1}}^2.
  \]
  Then~\eqref{eq:normstep} and Lemma~\ref{lem:uprod}  give
  \[
    \log \norm{v_n} \geq \frac{1}{2}    \eta \sum_{i=1}^n \inner{x_i, \wh{v}_{i-1}}^2 \geq \frac{1}{8}\sigma_1 - O(\sigma_2^2 \log^2 n \log \norm{v_n}),
  \]
  or
  \[
    \log \norm{v_n} \gtrsim \frac{\sigma_1}{1 + \sigma_2^2 \log^2 n}.
  \]
\end{proof}

\begin{claim}\label{claim:gaussianvecnorm}
  Let $a \sim N(0, 1)$.  For any two vectors $u$ and $v$, with probability $1-\delta$,
  \[
    \norm{au + v} \geq \delta \sqrt{\pi/2} \norm{u}.
  \]
\end{claim}
\begin{proof}
  First, without loss of generality $v$ is collinear with $u$; any
  orthogonal component only helps.  So we can only consider
  real-valued $u$ and $v$, and in fact rescale so $u = 1$.  The claim
  is then: with probability $1-\delta$, a sample from $N(v, 1)$ has
  absolute value at least $\delta\sqrt{\pi/2}$.  This follows from the
  standard Gaussian density being at most $1/\sqrt{2\pi}$.
\end{proof}

\thmOja*

\begin{proof}
  We assume that $\eta \norm{x_i}^2 \leq 1$ for all $i$, since the
  theorem is otherwise vacuous.

  We begin with the last statement.  Algorithm~\ref{alg:oja} only
  returns $\wh{v} \neq \perp$ if $s_n = \log \frac{\norm{v_n}}{\norm{v_0}} > 10 \log d$.
  But then by Lemma~\ref{lem:growth-gives-correctness},
  \[
    \norm{P\wh{v}_n} \leq \sqrt{\sigma_2} + \frac{\norm{v_0}}{\norm{v_n}} \leq  \sqrt{\sigma_2} + d^{-10}.
  \]

  All that remains is to show that, if $\sigma_1 > C \log d$ and
  $\sigma_2 < \frac{1}{C \log n}$, $\wh{v} \neq \perp$ with at least
  $1 - d^{-\Omega(C)}$ probability.  And of course,
  $\wh{v} \neq \perp$ if $\frac{\norm{v_n}}{\norm{v_0}} \geq d^{10}$.

  Oja's algorithm starts with $\wh{v}_0$ uniformly on the sphere, and
  is indifferent to the initial scale $\norm{v_0}$, so $v_0$ could be
  constructed as $\frac{v_0}{\norm{v_0}}$ for $v_0 \sim N(0, I_d)$.

  Let $v_0 = a v^* + u$ for $u \perp v^*$.  Let
  $B = \prod_{i=1}^n (I + \eta x_i x_i^T)$, so $v_n = B v_0$.

  By Lemma~\ref{lem:growth} and the bound on $\sigma_2$, we know
  $\norm{B v^*} \geq e^{c'\sigma_1}$ for some constant $c'$.  Then by
  Claim~\ref{claim:gaussianvecnorm}, with probability $1-\delta$,
  \[
    \norm{v_n} = \norm{aBv^* + Bu} \geq \delta \sqrt{\pi/2} \norm{B v^*} \geq \delta e^{c' \sigma_1}.
  \]
  The (very naive) Markov bound from $\E[\norm{v_0}^2] = d$ gives that
  \[
    \frac{\norm{v_n}}{\norm{v_0}} \geq \frac{\delta^{3/2} e^{c' \sigma_1}}{\sqrt{d}}
  \]
  with probability $1-2\delta$.  For sufficiently large $C$ in
  $\sigma_1 \geq C \log d$, this gives
  \[
    \frac{\norm{v_n}}{\norm{v_0}} \geq d^{10}
  \]
  with probability $1 - d^{-\Omega(C)}$.
\end{proof}

\subsection{Proof of Theorem~\ref{thm:upper}}\label{sec:upper}

\thmupper*
\begin{proof}
  Let $C$ be the constant in Theorem~\ref{thm:Oja1}.  For $R$ to be
  well defined, $\lambda_1 \neq 0$ so some $x_i \neq 0$.  Therefore
  $2^{-2b} \leq \lambda_1 \leq nd^2 2^{2b}$.  Thus one of the $\eta_i$
  considered in Algorithm~\ref{alg:unknownrate} is such that
  $\eta n \lambda_1 \in [C \log d, 2 C \log d]$.  Let $\wh{i}$ be this
  $i$.  For sufficiently large constant in the choice of $R$, we have
  \[
    \eta_i n \lambda_2 \leq \frac{1}{C \log n}
  \]
  for all $i \leq \wh{i}$.

  Let $\overline{x}$ be the $x_i$ of maximum norm, as computed by the
  algorithm.  We now show that
  $\wh{x} := \frac{\overline{x}}{\norm{\overline{x}}}$ is a sufficiently good answer if $\eta_{\wh{i}} \norm{\overline{x}}^2 \geq 1$.
  Decompose $\overline{x} = a v^* + b w$
  for $w \perp v^*$ a unit vector. By~\eqref{eq:otherdir}, $b$ is
  fairly small:

  \[
    b^2 \leq \eta_{\wh{i}} \sum_i \inner{x_i, w}^2 \leq \norm{PX^TXP} = n \lambda_2 \leq \frac{2 C \log d}{R \eta_{\wh{i}} }.
  \]
  The unit vector $\wh{x}$ in direction $\overline{x}$ has error
  \begin{align}
    \norm{P\wh{x}}^2 = \frac{b^2}{\norm{\overline{x}}^2} \leq \frac{2 C \log d}{R \eta_{\wh{i}} \norm{\overline{x}}^2} \lesssim \frac{\log d}{R \eta_{\wh{i}} \norm{\overline{x}}^2}.\label{eq:overbaraccuracy}
  \end{align}
  Therefore if $\eta_{\wh{i}} \norm{\overline{x}}^2 \geq 1$, $\wh{x}$ is a sufficiently accurate answer.

  The last statement in Theorem~\ref{thm:Oja1} shows that, if  $\eta_{i^*} \norm{\overline{x}}^2 \leq 1$ and $i^* \leq \wh{i}$, then
  \begin{align}
    \norm{Pv^{(i^*)}}^2 \leq (\sqrt{ \eta_{i^*} n \lambda_2} + d^{-9})^2 \lesssim \eta_{i^*} n \lambda_2 + d^{-18} \leq \frac{2 C \log d}{R} + d^{-18}\label{eq:OjaOutput}
  \end{align}
  which is sufficiently accurate.  We now split into case analysis.

  \textbf{In one case}, suppose $\eta_{\wh{i}}\norm{\overline{x}}^2 < 1$.
  Therefore the main body of Theorem~\ref{thm:Oja1} states that
  $v^{(\wh{i})} \neq \perp$ with high probability.  In particular,
  this means $i^* \leq \wh{i}$, so
  $\eta_{i^*} \norm{\overline{x}}^2 < 1$, and the algorithm's
  answer is $v^{(i^*)}$ which is sufficiently accurate by~\eqref{eq:OjaOutput}.

  \textbf{Otherwise}, $\eta_{\wh{i}}\norm{\overline{x}}^2 \geq 1$.
  Then outputting $\overline{x}$ is sufficiently accurate
  by~\eqref{eq:overbaraccuracy}.  If $i^* \geq \wh{i}$, the algorithm
  will definitely output $\overline{x}$; if $i^* < \wh{i}$, the
  algorithm might output $v^{(i^*)}$, but only if
  $\eta_{i^*} \norm{\overline{x}}^2 < 1$, in which case this is
  sufficiently accurate by~\eqref{eq:OjaOutput}.
\end{proof}

\subsection{Precision}\label{sec:precision}

Finally, we discuss why $O(\log(nd))$ bits of precision suffice for
the algorithm.  Algorithm~\ref{alg:oja} tracks two values: a unit
vector $\wh{v}_i$ and the log-norm $s_i$ of the unnormalized $v_i$.
The main concern is that the error in $\wh{v}_i$ could compound.

Consider $\wh{v}_i$ and $s_i$ to be the values computed by the
algorithm, which has some $\eps = \frac{1}{\poly(nd)}$ error (in
$\ell_2$) added in each iteration.  We can enforce $s_i \geq s_{i-1}$
despite the error.  Redefine $v_i$ to $2^{s_i} \wh{v}_i$.

We now redo the proof of~\eqref{eq:vwbound} with $\eps$ error in each step.
Define $B_i = \frac{\norm{v_i}^2}{\norm{v_0}^2} = 2^{s_i}$, and
$A_i = \log \frac{B_i}{B_{i-1}} = (s_i - s_{i-1})$ which satisfies
$A_i \geq \eta \inner{x_i, \wh{v}_{i-1}}^2 - O(\eps)$
by~\eqref{eq:normstep}.  Therefore
\begin{align}
  \eta\sum_{i=1}^n \inner{x_i, v_{i-1}}^2 \leq \norm{v_0}^2 \sum_{i = 1}^n (A_i + O(\eps)) B_{i-1} \leq \norm{v_0}^2 ((B_n-1) + O(\eps n B_n)) = \norm{v_n}^2 - \norm{v_0}^2 + O(\eps n \norm{v_n}^2)
  \label{eq:normthingprec}
\end{align}
by Claim~\ref{claim:prodab}.  Then for any unit vector $w$ with $Pw = w$,
\begin{align*}
  \inner{v_n - v_0, w}^2 &= (\sum_{i=1}^n \inner{v_i - v_{i-1}, w})^2\\
                         &= (\sum_{i=1}^n \eta \inner{x_i, v_{i-1}} \inner{x_i, w} + O(\eps) \norm{v_{i-1}})^2\\
                         &= \left(O(n \eps \norm{v_n}) + \eta \sum_{i=1}^n \inner{x_i, v_{i-1}} \inner{x_i, w}\right)^2 && \text{by~\eqref{eq:unnormalized-update}}\\
                         &\leq \eta \sum_{i=1}^n \inner{x_i, v_{i-1}}^2  \cdot \eta \sum_{i=1}^n\inner{x_i, w}^2 + O(n^2 \eps \norm{v_n}^2)&& \text{by Cauchy-Schwarz}\\
                         &\leq (\norm{v_n}^2 - \norm{v_0}^2) \sigma_2 + O(\eps n^2 \norm{v_n}^2). &&\text{by~\eqref{eq:normthingprec} and \eqref{eq:otherdir}}
\end{align*}
There's nothing special about the start and final indices, giving the following bound for general indices
$a \leq b$:
\begin{align}
  \inner{v_b - v_a, w}^2 &\leq (\norm{v_b}^2 - \norm{v_a}^2) \sigma_2 + O(\eps n^2 \norm{v_b}^2)\label{eq:vwbound-prec}.
\end{align}

Given~\eqref{eq:vwbound-prec}, the error tolerance flows through the
rest of the proof easily.  Lemmas~\ref{lem:growth-gives-correctness}
and~\ref{lem:normalized-movement} follow immediately with
$O(\eps n^2)$ additive error.  Lemma~\ref{lem:uprod} gets additive
error $O(\sigma_2 \eps n^3 \log^2 n)$, so both the numerator and
denominator of Lemma~\ref{lem:growth} change by $\eps \poly(n)$.  Both
the conditions and result of Theorem~\ref{thm:Oja1} only change by an
additive $\eps \poly(n)$ error, which for sufficiently small
polynomial $\eps$ are absorbed by the constant factors and
$\frac{1}{d^9}$ additive error.  And Algorithm~\ref{alg:unknownrate}
does nothing that could compound the error by more than a constant
factor, so Theorem~\ref{thm:upper} holds as well.

\section{Lower Bound for Mergeable Summaries}
In this section, we show that any all the mergeable summaries require $\Omega(d^2 / R^2)$ bits of space, even just to approximate PCA with 0.1 error. This is significantly worse than our upper bound for streaming algorithms.
\mergeablelowerbound*
    To prove the theorem, for $p > 1$, we define distribution $\mathcal{D}_p$ over $\R^{d \times d}$ such that $X \sim \mathcal{D}_p$ is drawn according to the following randomized procedure:
    \begin{enumerate}
        \item Sample $v^* \sim \mathcal{N}(0, I_d)$.
        \item Define $k := d / p$. For each $i \in [p]$, we sample a $X^{(i)} \in \R^{k \times d}$ as follows: Randomly choose a $j_i^* \in [k]$ and set $X^{(i)}_{j_i^*}$ to be $v^{*}$.
        For $j \in [k] \setminus \{j_i^*\}$, independently sample $X^{(i)}_{j} \sim \mathcal{N}(0, I_d)$. Let $X^{(i)} := (X^{(i)}_1, X^{(i)}_2, \dots, X^{(i)}_k)^T$.
        \item Finally, let $X$ be the concatenation of all the $X^{(i)}$'s, i.e., \[
        X := \begin{pmatrix}
    X^{(1)} \\
    X^{(2)} \\
    \vdots \\
    X^{(p)}
    \end{pmatrix}.
        \]
    \end{enumerate}

We first show that for $X \sim \mathcal{D}_p$, $X^TX$ has a large spectral gap with high probability:

\begin{lemma}
    \label{lem:sketching_spectral_gap}
    For $X \sim \mathcal{D}$,  with $1 - o(1)$ probability, \[
        \min_{\substack{v' \perp v^*\\\norm{v'} = 1}} \frac{\|X\wh{v^*}\|^2} {\norm{Xv'}^2} \ge 0.1p.
    \]
    Furthermore, the spectral gap of $X^TX$ is at least $0.1p$.
\end{lemma}

\begin{proof}
    We can decompose the rows of $X$ into two parts: repetitions of $v^*$ and other randomly sampled rows. Define $\wt{X} \in \R^{(d - p) \times d}$ as $X$ excluding the  $v^*$'s in each $X^{(i)}$. 
    
    For an arbitrary unit vector $u \in \R^d$, $X$'s variance on $u$ is equal to \[
        \|X u\|^2 = p \inner{v^*, u}^2 + \|\wt{X}u\|^2.
    \]
    Therefore, $X$'s variance at the direction of $v^*$ is  \[
        \|X \wh{v^*}\|^2 = p \inner{v^*, \wh{v^*}}^2 + \|\wt{X} \wh{v^*}\|^2 \ge  p \inner{v^*, \wh{v^*}}^2 = p\norm{v^*}^2.
    \]
     For any unit vector $u$ orthogonal to $v^*$, we have \[
        \|X u\|^2 = p \inner{v^*, u}^2 + \|\wt{X}u\|^2 = \|\wt{X}u\|^2 \le \|\wt{X}\|^2.
     \]
     Utilizing this, we have \[
       \min_{\substack{v' \perp v^*\\\norm{v'} = 1}} \frac{\|X\wh{v^*}\|^2} {\norm{Xv'}^2} \ge \frac{p\|v^*\|^2}{\|\wt{X}\|^2}.
     \]

      Note that $v^* \sim \mathcal{N}(0, I_d)$, by Lemma~\ref{lem:chi_squared_concentration}, $\norm{v^*}^2 \ge 0.9d$ holds with probability at least $1 - o(1)$.
      In addition, since every instance in $\wt{X}$ follows $\mathcal{N}(0, 1)$ independently, by Lemma~\ref{lem:rudelsonvershynin}, we have \[
        \Pr\left[\|\wt{X}\| \ge 3\sqrt{d}\right] \le o(1).
      \]
      This states that with probability $1 - o(1)$, \[
        \min_{\substack{v' \perp v^*\\\norm{v'} = 1}} \frac{\|X\wh{v^*}\|^2} {\norm{Xv'}^2} \ge \frac{p\|v^*\|^2}{\|\wt{X}\|^2} \ge \frac{0.9pd}{9d} = 0.1p.
      \]
      Furthermore, the spectral gap of $X^TX$ is given by\[
       \max_{\norm{v} = 1} \min_{\substack{v' \perp v\\\norm{v'} = 1}} \frac{\norm{Xv}^2}{\norm{Xv'}^2} \ge \min_{\substack{v' \perp v^*\\\norm{v'} = 1}} \frac{\|X\wh{v^*}\|^2} {\norm{Xv'}^2} \ge 0.1p.
      \]
\end{proof}

\begin{lemma}
    \label{lem:v*_close_to_PCA}
    For $X \sim \mathcal{D}$, let $v$ be the top eigenvalue of $X^TX$. Let $\alpha$ be an arbitrarily small positive constant. There exists a constant $C$ such that when $p \ge C$, with $1 - o(1)$ probability, $\sin^2(\wh{v^*}, v) \le \alpha^2$.
\end{lemma}
\begin{proof}
    Without loss of generality, we express $v$ as $
        v = \sqrt{1 - \eps^2} \wh{v^*} + \eps u$
    for some $\eps > 0$ and unit vector $u$ orthogonal to $v^*$.
    We only need to prove that $\eps \le \alpha$.

    We have
    \begin{align*}
        \norm{Xv}^2 \le& (1 - \eps^2)\nnorm{X\wh{v^*}}^2 + \eps^2\norm{Xu}^2 + 2\sqrt{1 - \eps^2}\eps \nnorm{X\wh{v^*}}\cdot \norm{Xu}.
    \end{align*}
    By \cref{lem:sketching_spectral_gap}, we have with probability $1 - o(1)$, $\nnorm{X\wh{v^*}}^2 \ge 0.1p\nnorm{Xu}^2$. Therefore, \[
        \norm{Xv}^2 \le \nnorm{X\wh{v^*}}^2(1 - \eps^2 + \eps^2\frac{10}{p} + 2\eps\sqrt{\frac{10}{p}}).
    \]
    When $\eps > \alpha$, there exists a constant $C > 0$ such that for $p \ge C$, \[
         \nnorm{X\wh{v^*}}^2(1 - \eps^2 + \eps^2\frac{10}{p} + 2\eps\sqrt{\frac{10}{p}}) \le \nnorm{X\wh{v^*}}^2(1 - \frac{\eps^2}{2}) \le \nnorm{X\wh{v^*}}^2.
    \]
    This contradicts the assumption that the direction of $v$ has larger variance than $v^*$. This proves the lemma.
\end{proof}

Let $\mathcal{A}$ be an arbitrary deterministic mergeable summary for PCA that satisfies 
\begin{equation}
    \label{eq:sketching_succeeds}
    \Pr_{X \sim \mathcal{D}_p}[\sin^2(\mathcal{A}(X), v^*) \le 0.105] \ge 0.9.    
\end{equation}

\begin{lemma}
    \label{lem:A_lower_bound}
    $\mathcal{A}$ requires $\Omega(d^2 / p^2)$ bits of space.
\end{lemma}
Before proving Lemma~\ref{lem:A_lower_bound}, we first show the proof of Theorem~\ref{thm:sketching} assuming Lemma~\ref{lem:A_lower_bound} is true.

\begin{proof}[Proof of Theorem~\ref{thm:sketching}]
    Suppose we have a mergeable summary $\mathcal{S}$ that 0.1-approximates PCA uses $o(d^2 / R^2)$ bits of space with high probability.
    Let $\alpha > 0$ be a constant such that \[
        \sin^2(\arcsin \alpha + \arcsin \sqrt{0.1}) \le 0.105.
    \]
    Let $C$ be the constant in \cref{lem:v*_close_to_PCA} corresponding to $\alpha$.
    By Lemma~\ref{lem:sketching_spectral_gap}, for $X \sim \mathcal{D}_{10R + C}$, $X^TX$ has spectral gap $R$ with high probability. Therefore, $\mathcal{S}$ succeeds in 0.1-approximating PCA with high probability. Combining with \cref{lem:v*_close_to_PCA}, we have $\mathcal{S}$ succeeds in 0.105-approximating $v^*$ on $X \sim \mathcal{D}_{10R+C}$ with high probability.
    
    By Yao's minimax principle, there must be a \textit{deterministic} mergeable summary $\mathcal{A}$ that also uses $o(d^2 / R^2)$ bits of space and 0.105-approximates PCA on $X \sim \mathcal{D}_{10R + C}$ with high probability, i.e., it satisfies \eqref{eq:sketching_succeeds}. By Lemma~\ref{lem:A_lower_bound},  $\mathcal{A}$ must require $\Omega(d^2 / R^2)$ bits of space, which is a contradiction.
    
\end{proof}

We use $s$ to denote the bits of space of $\mathcal{A}$. To prove Lemma~\ref{lem:A_lower_bound}, we will show that $s = \Omega(d^2 / p^2)$. We use $m_i$ to denote $\mathcal{A}$'s summary for $X^{(i)}$. The key property we use here is that each $m_i$ is a deterministic function of $X^{(i)}$, so $m_i$'s are independent except for the shared vector $v^*$. We start with the following classical result:
\begin{proposition}[Chain Rule for Mutual Information]
\label{prop:conditional_chain_rule}
    Let $X$, $Y$, $Z$ be random variables. We have\[
        I(X; Y \mid Z) = I(X; Y) - (I(X; Z) - I(X; Z \mid Y)).
    \]
\end{proposition}
 
\begin{corollary}
\label{cor:independent_mutual_information}
     Let $X$, $Y$, $Z$ be random variables. If $X$ and $Z$ are independent, \[
        I(X; Y) \le I(X; Y \mid Z).
     \]
\end{corollary}
\begin{proof}
    Since $X$ and $Z$ are independent, $I(X; Z) = 0$. Applying proposition~\ref{prop:conditional_chain_rule} gives the result.
\end{proof}

Using these results, we can prove the next two lemmas  bounding the mutual information between $v^*$ and $m_i$'s in terms of $s$.

\begin{lemma}
    \label{lem:m_i_v_star}
    For every $i \in [p]$, $I(m_i; v^*) \le s/k$.
\end{lemma}
\begin{proof}
    Since $m_i$ and $j_i^*$ are independent, by corollary~\ref{cor:independent_mutual_information} we have \[
         I(m_i; X_{j_i^*}^{(i)}) \le I(m_i; X_{j_i^*}^{(i)} \mid j_i^*) .
    \]
    Thus, \[
     I(m_i; v^*) = I(m_i; X_{j_i^*}^{(i)}) \le I(m_i; X_{j_i^*}^{(i)} \mid j_i^*) = \frac{1}{k}\sum_{j \in [k]} I(m_i; X_{j}^{(i)} \mid j_i^* = j) = \frac{1}{k}\sum_{j \in [k]} I(m_i; X^{(i)}_j).
    \]
    Furthermore, since each $X^{(i)}_j$ is sampled independently, by applying corollary~\ref{cor:independent_mutual_information}, we have for each $j \in [k]$, \[
        I(m_i; X^{(i)}_j) \le I(m_i; X^{(i)}_j \mid X^{(i)}_1, \dots, X^{(i)}_{j-1}).
    \]
    We have \[
        \sum_{j \in [k]} I(m_i; X^{(i)}_j) \le \sum_{j \in [k]} I(m_i; X^{(i)}_j \mid X^{(i)}_1, \dots, X^{(i)}_{j-1}) = I(m_i; X^{(i)}_1, \dots, X^{(i)}_{n}) \le H(m_i).
    \]
    Since $\mathcal{A}$ only has $s$ bits of space, $H(m_i) \le s$.
    Therefore, \[
        I(m_i; v^*) \le \frac{1}{k} H(m_i) \le \frac{s}{k}.
    \]
\end{proof}
\begin{lemma}
    \label{lem:component_information}
    $I(v^* ; m_1, m_2, \dots, m_p) \le p^2 s / d$ 
\end{lemma}
\begin{proof}
    We first prove that for $i \in [p]$, $I(v^* ; m_i \mid m_1, \dots, m_{i-1}) \le I(v^*; m_i)$.           Note that \[
        I(m_i; m_1, \dots, m_{i-1} \mid v^*) \le I(X^{(i)}; X^{(1)}, \dots, X^{(i-1)} \mid v^*) = 0.
    \]
    By proposition~\ref{prop:conditional_chain_rule}, we have 
    \begin{align*}
        I(v^* ; m_i \mid m_1, \dots, m_{i-1}) = I(v^*; m_i) - I(m_i; m_1, \dots, m_{i-1}) + I(m_i; m_1, \dots, m_{i-1} \mid v^*) \le I(v^*; m_i).
    \end{align*}
    Then by Lemma~\ref{lem:m_i_v_star}, we have \[
        I(v^*; m_1, \dots, m_p) = \sum_{i \in [p]}  I(v^* ; m_i \mid m_1, \dots, m_{i-1}) \le \sum_{i \in [p]} I(v^*; m_i) \le \frac{ps}{k} = \frac{p^2s}{d}.
    \]
\end{proof}

Next, we show that the mutual information between $v^*$ and the output of $\mathcal{A}$ must be at least $\Omega(d)$. For this purpose, we refer to a special case of lemma 4.4 from \cite{jalal2021instance}:

\begin{lemma}[Lemma 4.4 of \cite{jalal2021instance}] 
\label{lem:estimate_info}
Consider random variable $x$ uniformly distributed over $D \subseteq \R^d$ and random variable $\tilde{x}$ in $\R^d$.
If the joint distribution of  $ (x, \tilde{x}) $ satisfies 
\[
\Pr[\|x - \tilde{x}\| \leq \eta] \geq 0.9,
\]
then we have
\[
\frac{1}{8}\log \mathrm{Cov}_{3\eta, 1/2}(D) \leq I(x; \tilde{x}) + 1.98,
\]
where $\mathrm{Cov}_{3\eta, 1/2}$ denotes the minimum number of $d$-dimensional balls of radius $3\eta$ required to cover at least half of $D$.
\end{lemma}

\begin{lemma}
    \label{lem:output_information}
    Let unit vector $\tilde{v}$ be an approximation of $\wh{v^*}$ such that \[
        \Pr[\sin^2(v^*, \tilde{v}) \leq 0.105] \geq  0.9.
    \]
    Then,  \[
        I(\wh{v^*}; \tilde{v}) \gtrsim d.
    \]
\end{lemma}
    
\begin{proof}
    We define $\tilde{v}' := \sign(\tilde{v}^Tv^*)\tilde{v}$. Then it is easy to verify that $\sin^2(v^*, \tilde{v}) \le 0.105$ implies that $\|\wh{v^*} - \tilde{v}'\| \le 1/3 - c$ for some constant $c > 0$. Therefore, by Lemma~\ref{lem:estimate_info}, we have \[
        I(\wh{v^*} ; \tilde{v}') \ge \frac{1}{8} \log \mathrm{Cov}_{1 - 3c, 1/2}(S_d) - 1.98,
    \]
    where $S_d$ denotes the $d$-dimensional unit sphere. Note that each ball of radius $1 - 3c$ can cover a spherical cap with height at most $1 - 3c$ on the unit sphere, and the union of these caps need to cover at least half of the surface area of a unit sphere. Using a bound on the area of a spherical cap (Lemma~\ref{lem:spherical_cap_area}), we have \[
        \log \mathrm{Cov}_{1 - 3c, 1/2} \ge \log \frac{\text{area of }S_d}{\text{area of height-$(1 - 3c)$ spherical cap}} \gtrsim d.
    \]
    Therefore, \[
         I(\wh{v^*} ; \tilde{v}') \gtrsim  d.
    \]
    This implies that  \[
        d \lesssim I(\wh{v^*}; \tilde{v}') = I(\wh{v^*}; \sign(\tilde{v}^Tv^*)\tilde{v}) \le I(\wh{v^*}; \tilde{v}, \sign(\tilde{v}^Tv^*)) = I(\wh{v^*}; \tilde{v}) + I(\wh{v^*}; \sign(\tilde{v}^Tv^*) \mid \tilde{v}).
    \]
    In addition, since \[
        I(\wh{v^*}; \sign(\tilde{v}^Tv^*) \mid \tilde{v}) \le H(\sign(\tilde{v}^Tv^*)) \le 1,
    \]
    we have \[
        I(\wh{v^*}; \tilde{v}) \gtrsim d.
    \]
\end{proof}

This gives us a lower bound for $s$:
\begin{proof}[Proof of Lemma~\ref{lem:A_lower_bound}]  
    Let $\tilde{v}$ be the output of $\mathcal{A}$.
    By Lemma~\ref{lem:output_information}, we have \[
        I(\tilde{v}; \wh{v^*}) \gtrsim d.
    \]
    By Lemma~\ref{lem:component_information},  \[
        I(m_1, \dots, m_p; v^*) \le \frac{p^2s}{d}.
    \] Using the data processing inequality, we get \[
       d \lesssim I(\tilde{v}; \wh{v^*}) \le I(m_1, \dots, m_p; v^*) \le \frac{p^2s}{d}.
    \]
    Therefore,  \[
        s \gtrsim \frac{d^2}{p^2}.
    \]
\end{proof}

\section{Lower Bound for Accuracy}

Our lower bound is based on the \textsc{PartialDuplicate} instance, where an instance is a matrix $X \in \{0, -1, 1\}^{(k + n + 1) \times d}$ can be expressed as follows:
\begin{itemize}
\item The first row equals $x + y$, where  $x, y \in \{0, -1, 1\}^d$ have $\supp(x) = \{1, 2, \dotsc, d / 2\}$ and $\supp(y) = \{d / 2 +1, \dotsc, d\}$.
\item For $i \in \{2, \dots, k+1\}$, the $i$-th row equals $x$.
\item The last $n$ rows form a uniformly random matrix $X' \in \{-1, 1\}^{n \times d}$.
\end{itemize}
That is, the entries look like:
\[
  X =
  \begin{array}{|c|c|}
    \hline
    {x}& y\\
    {x}& 0\\
    \vdots&\vdots\\
    {x}& 0\\
    \hline
    \multicolumn{2}{|c|}{X'}\\
    \hline
  \end{array}
\]
except that $x, y$ are zero-padded to $d$ dimensions. Without loss of generality, we assume $d$ is superconstant and $k = o(d)$.

\subsection{Spectral properties of \textsc{PartialDuplicate}}

Let $v^*$ be the top unit eigenvector of $X^TX$.
We can decompose $v^*$ into three components: the
${x}$ direction, the ${y}$ direction, and the
component orthogonal to both of these.  This is:
\[
 v^* = a \wh x + b \wh y + c \wt w,
\]
where $a^2 + b^2 + c^2 = 1$ and $\wt w$ is an arbitrary unit vector orthogonal to
${x}$ and $y$. 

We have that
\begin{align}
    \|Xv^*\|^2 &= \|X'v^*\|^2 + k \inner{x, v^*}^2 + \inner{x + y, v^*}^2 \nonumber\\ 
    &=  \|X'v^*\|^2 + a^2k\|x\|^2 + (a \|x\| + b \|y\|)^2 \nonumber\\
    &= \|X'v^*\|^2 + a^2(k + 1)\|x\|^2 + 2ab \|x\| \|y\| + b^2\|y\|^2\nonumber\\
    &= \|X'v^*\|^2 + \frac{a^2(k + 1) d}{2} + abd +\frac{b^2 d}{2}  \label{eq:Xw}.
\end{align}

\begin{lemma}
\label{lem:c_equals_0}
    Suppose $n \le d$. Then $|c| \le O(1 / k)$ with high probability.
\end{lemma}

\begin{proof}
    Since $\sign(b) \cdot v^*$ is also a top eigenvector of $X^TX$, without loss of generality, we assume $b > 0$.
    We consider unit vector $v' = \sqrt{a^2 + c^2} \wh x + b \wh y$, we have 
    \begin{align*}
        &\|Xv'\|^2 - \|Xv^*\|^2 \\
        =& \|X'v'\|^2  - \|X'v^*\|^2 +  \frac{c^2(k + 1) d}{2} + (\sqrt{a^2 + c^2} - a)bd \\
        \ge & (v' + v^*)^T{X'}^TX'(v' - v^*) + \frac{c^2(k + 1) d}{2}.
    \end{align*}
    By Lemma~\ref{lem:rudelsonvershynin}, with high probability, \[
        \abs{(v' + v^*)^T{X'}^TX'(v' - v^*)} \le  \norm{X'}^2 \norm{v' - v^*} \norm{v' + v^*} \le O(\abs{c d}).
    \]
    Thus, \[
        \|Xv'\|^2 - \|Xv^*\|^2 \ge  - O(\abs{c d}) +  \frac{c^2(k + 1) d}{2}.
    \]
    Since $v^*$ is the top eigenvector, this implies that \[
        - O(\abs{cd}) +  \frac{c^2(k + 1) d}{2} \le 0.
    \]
    Hence, \[
        \abs{c} \le O(1/k).
    \]
\end{proof}

\begin{lemma}
    \label{lem:b_is_large}
    Suppose $k \ge C$ and $n \le \frac{d}{9k}$ for a sufficiently large constant $C$. Then $|b| \ge \frac{1}{3k}$ with high probability.
\end{lemma}
\begin{proof}
    Suppose $|b| < \frac{1}{3k}$ with non-negligible probability.  Combining with Lemma~\ref{lem:c_equals_0}, this implies that with non-negligible probability, $|b| < \frac{1}{3k}$ and $|c| \le O(1/k)$. We will show that with high probability, any unit vector $v = a \wh x + b \wh y + c \wt w$ satisfying $|b| < \frac{1}{3k}$ and $|c| < \frac{\log k}{k}$ is not the top eigenvector of $X^TX$. This contradicts the assumption and proves the lemma.

    Without loss of generality, we only consider the case when $b > 0$. Let $t := \sqrt{a^2 + b^2} = \sqrt{1 - c^2} > 2 / 3$. 
    Therefore, we have $a = \sqrt{t^2 - b^2} = t - \Theta(b^2)$. 
    Taking this into \eqref{eq:Xw}, we have \[
        \norm{Xv}^2 = \|X'v\|^2 + \frac{(t^2 - b^2)k d}{2} + b(t - \Theta(b^2))d.
    \]
    We consider vector $v' = \sqrt{t^2 - (b + \eps)^2} \wh x + (b + \eps) \wh y + c\wt w$ for $\eps = \frac{1}{3k}$.
    Now we prove that with high probability $\|Xv'\|^2 - \|Xv\|^2 > 0$.
    We have
    \begin{align*}
       \norm{Xv'}^2 - \norm{Xv}^2 &= \norm{X'v'}^2 - \norm{X'v}^2 + \frac{kd}{2}(b^2 - (b + \eps)^2) + \eps td + \Theta(b^3)d - \Theta((b + \eps)^3)d\\
       &= \norm{X'v'}^2 - \norm{X'v}^2 - bkd\eps - \frac{kd\eps^2}{2} +\eps td + \Theta(b^3)d - \Theta((b + \eps)^3)d \\
       &\ge \norm{X'v'}^2 - \norm{X'v}^2 - \frac{d}{9k} - \frac{d}{18k} + \frac{d}{3k} \pm O(\frac{1}{k^3})d\\
       &\ge - \norm{X'v}^2 + \frac{d}{8k}
    \end{align*}
    with our choice of $k$.
   Note that
    \begin{align*}
        \norm{X'v}^2 \le \norm{aX' \wh{x}}^2 + \norm{bX' \wh{y}}^2 + \norm{c X' \wt w}^2 &\le \norm{X'\wh{x}}^2 + \frac{1}{9k^2}\norm{X'\wh{y}}^2 + \frac{\log^2 k}{k^2} \norm{X'}^2.
    \end{align*}
    By Lemma~\ref{lem:rudelsonvershynin}, with high probability, \[
        \norm{X'} \le 2 \sqrt{d}. 
    \]
    Furthermore, since $\wh{x}$ and $\wh{y}$ are independent of $X'$, by Claim~\ref{claim:subgamma-sum}, with high probability, \[
        \norm{X'\wh x}^2 \le n + o(n) \quad \text{and} \quad \norm{X' \wh y}^2 \le n + o(n).
    \]
    Therefore, with high probability, \[
        \norm{X'v}^2 \leq 1.1 n < \frac{d}{8k}.
    \]
    Hence, with high probability, \[
        \norm{Xv'}^2 - \norm{Xv}^2 > 0.
    \]

\end{proof}

\begin{lemma}
    \label{lem:approx_direction}
    Suppose $k \ge C$, $n \le \frac{d}{9k}$ and $\eps \le \frac{1}{Ck^2}$ for a sufficiently large constant $C$. For any $\eps$-approximate PCA solution $w$, ${\inner{w, y}} \ge \Omega(\sqrt{d}/{k})$ with high probability.
\end{lemma}
\begin{proof}
    By Lemma~\ref{lem:b_is_large}, with high probability, \[
        {\inner{v^*, y}} = \inner{b \wh y, y} \ge \Omega(\sqrt{d} / k) .
    \]
    Therefore, for $w = v^* + \sqrt{\eps} u$ for some unit vector $u$, we have \[
        \inner{w, y} = \inner{v^* + \sqrt{\eps} u, y}  = \inner{v^*, y} + \sqrt{\eps}\inner{u, y} \ge \Omega(\sqrt{d} / k) - \sqrt{\eps} \norm{y} \ge \Omega(\sqrt{d} / k).
    \]
\end{proof}

\begin{lemma}\label{lem:spectralgap}
  Suppose $n \le d$. The spectral gap $R$ is at least $k / 20$ with high probability.
\end{lemma}
\begin{proof}
The first eigenvalue $\lambda_1$ of $X^TX$ satisfies
  \begin{align*}
    \lambda_1 &= \max_{\norm{v} = 1} \norm{Xv}^2 \geq \norm{X\wh x}^2 \geq (k + 1) \inner{x, \wh x}^2 \ge \frac{kd}{2}.
  \end{align*}
  The second eigenvalue $\lambda_2$ of $X^TX$ satisfies
  \begin{align*}
    \lambda_2 &= \min_v \max_{\substack{v' \perp v\\\norm{v'} = 1}} \norm{Xv'}^2\\
             &\leq \max_{\substack{v' \perp {x}\\\norm{v'} = 1}} \norm{Xv'}^2\\
               &= \max_{\substack{v' \perp {x}\\\norm{v'} = 1}} \left(\norm{X' v'}^2 + \inner{y, v'}^2 \right)\\
               &\leq \norm{X'}^2 + \frac{d}{2}.
  \end{align*}
  By Lemma~\ref{lem:rudelsonvershynin} , with high probability, $\|X'\| \le 3\sqrt{d}$. Therefore, \[
    \lambda_2 \leq 10d.
  \]
  Hence the spectral ratio
  \[
    R = \frac{\lambda_1}{\lambda_2} \ge \frac{k}{20}.
  \]
\end{proof}

\subsection{Accuracy lower bound}
\lowerbound*

\begin{proof}
  Suppose that we have such an $\eps$-approximate streaming PCA
  algorithm.  We set up a two player one-way communication protocol.
  Let $A_1 \in \{-1, 1\}^{n \times \frac{d}{2}}$ and
  $A_2 \in \{-1, 1\}^{n \times \frac{d}{2}}$ be chosen uniformly at
  random.  Let $A = [A_1, A_2] \in \{-1, 1\}^{n \times d}$ be their
  concatenation. Let $A' = [A_1, 0] \in \{0, -1,  1\}^{n \times d}$ be the matrix that pads $A_1$ to $d$ columns with $0$.

  In this protocol, Alice receives $A = [A_1, A_2]$ and Bob receives
  $A_1$.  Alice feeds $A$ to the streaming algorithm, reaching some
  stream state $S$, which she sends to Bob.  Bob uses $A_1$ and $S$ to
  construct an approximation $\wh{A}$ to $A_2$ in the following
  fashion.  For each $i \in [n]$, 
  Bob sets the streaming algorithm's state to $S$, inserts the $i$-th row of $A'$ for $k$ times and computes the algorithm's approximate PCA
  solution $\wh{v}_i$. Let $\wh{V} \in \R^{n \times d}$ be the matrix with the $i$-th row being $\wh{v}_i$.  Let
  $\wh{V}_2 \in R^{n \times \frac{d}{2}}$ be the last $d / 2$ columns
  of $\wh{V}$.  We will show that $I(A_2; \wh{V}) \gtrsim d^2 / R^3$ for
  an appropriate choice of parameters.

  Note that when Bob produces $\wh{v}_i$, the streaming algorithm has
  effectively seen the stream $A$ followed by $k$ vectors that match
  the $i$th row of $A$.  Up to reordering of rows, this is distributed
  identically to \textsc{PartialDuplicate}.  Reordering
  the rows, of course, does not change the covariance matrix.

  We choose $k = \max(20R, C)$, $n = \frac{d}{9k}$ and $\eps =\frac{1}{Ck^2}$ for the
  constant $C$ in Lemma~\ref{lem:approx_direction}.
  By Lemma~\ref{lem:spectralgap}, with high probability the stream has
  spectral gap at least $k/20 \geq R$.  Therefore the streaming
  algorithm's PCA solution should be $\eps$-approximate with at least
  $2/3$ probability. Then
  Lemma~\ref{lem:info-inequality} says that
  \[
    I(\wh{V}; A_2) \geq \Omega\left( \frac{1}{k^2} \cdot \frac{d}{k} \cdot \frac{d}{2}\right) - d  = \Omega\left({d^2}/{R^3}\right).
     \]
  Now, $\wh{V}$ is independent of $A_2$ conditioned on $(S, A_1)$ so by the data processing inequality,
  \[
    I(\wh{V}; A_2) \leq I(A_1, S; A_2) \leq I(A_1; A_2) + I(S; A_2 \mid A_1) \leq 0 + H(S).
  \]
  Thus, if the state $S$ contains $\abs{S}$ bits, we have
  \begin{align*}
    \Omega(d^2/R^3) \leq H(S) = H(\abs{S}) + H(S \mid \abs{S}) \leq \E[\abs{S}] + H(\abs{S})
  \end{align*}
  Now, for any random variable $X$ over positive integers,
  \begin{align*}
    H(X) &= \sum_{i=1}^{\infty} p(i) \log \frac{1}{p(i)}\\
         &= \left(\sum_{i: p(i) \leq 2^{-i}} p(i) \log \frac{1}{p(i)}\right) + \left(\sum_{i: p(i) > 2^{-i}} p(i) \log \frac{1}{p(i)}\right)\\
         &\leq \left(\sum_{i: p(i) \leq 2^{-i}} 2^{-i} \cdot i \right) + \left(\sum_{i: p(i) > 2^{-i}} i p(i) \right)\\
         &= 2 + \E[X]
  \end{align*}
  so $\Omega(d^2/R^3) \leq 2\E[\abs{S}] + 2$, or
  \[
    \E[\abs{S}] \geq \Omega(d^2/R^3).
  \]
  Thus the streaming algorithm must store $\Omega(d^2/R^3)$ bits on
  average after Alice has finished feeding in her part of the stream.
\end{proof}

\section*{Acknowledgments}
We thank David Woodruff and anonymous reviewers for helpful comments.
Eric Price and Zhiyang Xun are supported by NSF award CCF-1751040 (CAREER) and the NSF AI Institute for Foundations of Machine Learning (IFML).

\printbibliography

\appendix

\section{Utility lemmas for the Lower Bounds}
We use the following bound on the maximum singular value of an iid
subgaussian matrix:

\begin{lemma}[Feldheim and Sodin~\citep{feldheim2010universality}, see also (2.4)
        of~\cite{rudelson2010non}] \label{lem:rudelsonvershynin} Let
  $A$ be an $n \times N$ random matrix with independent
  subgaussian entries of zero mean and variance $1$, for $n \leq N$.
  There exists a universal constant $c > 0$ such that
  \[
    \Pr[\norm{A} \geq \sqrt{n} + \sqrt{N} + \tau \sqrt{N}] \lesssim e^{-c n \tau^{3/2}}
  \]
  for any $\tau > 0$.
\end{lemma}

The following is essentially a restatement of the JL lemma for $\pm 1$
matrices:
\begin{claim}\label{claim:subgamma-sum}
  Let $u\in \R^d$ be a unit vector, and $X \in \{-1, 1\}^{n \times d}$
  independently and uniformly.  Then
  \[
    \E[\norm{Xu}^2] = n
  \]
  and with $1-\delta$ probability
  \[
    \abs{\norm{Xu}^2 - n} \lesssim \sqrt{n \log
      \frac{1}{\delta}}  + \log \frac{1}{\delta}.
  \]
\end{claim}
\begin{proof}
  Let $z = Xu$.  The coordinates $z_i$ are independent,
  mean zero, variance $1$, and subgaussian with variance parameter
  $1$.  The expectation bound is trivial: sum the variance over $n$
  independent coordinates.  For concentration, each coordinate $z_i^2$
  is a squared subgaussian, and hence subgamma with $(\sigma, c)$
  parameters $(O(1), O(1))$.  Then $\sum_i z_i^2$ is subgamma with
  parameters $(O(\sqrt{n}), O(1))$.  Hence with probability
  $1-\delta$ we have
  \[
    \abs{\norm{Xu}^2 - n} \lesssim \sqrt{n \log \frac{1}{\delta}} +
    \log \frac{1}{\delta}.
  \]
\end{proof}

\begin{lemma}\label{lem:info-inequality}
  Let $X \in \{-1, 1\}^{n \times d}$ be uniformly distributed, and let
  $Y \in \R^{n \times d}$ have rows of norm at most $1$ such that each row
  $i \in [n]$ has $\abs{\inner{x_i, y_i}} > a \sqrt{d}$ with at least
  $50\%$ probability, for $a > 0$.  Then
  \[
    I(X; Y) \geq \Omega( a^2 nd) - n.
  \]
\end{lemma}
\begin{proof}
  For any row $y$, when $x \in \{-1, 1\}^d$ uniformly at
  random, $\inner{x, y}$ is subgaussian with variance parameter $\norm{y}^2 \leq 1$, so
  \[
    \Pr[\abs{\inner{x, y}} > a \sqrt{d}] \leq 2 e^{-a^2 d/2},
  \]
  so the number of $x$ with $\abs{\inner{x, y}} > a \sqrt{d}$ is at
  most $2^{(1 - \Omega(a^2))d}$.  Let $b \in \{0, 1\}^n$ denote the
  indicator vector with $b_i = 1$ if
  $\abs{\inner{x_i, y_i}} > a \sqrt{d}$ and $b_i = 0$ otherwise.

  For any $Y, b$, let $S_{Y,b} \subseteq \{-1, 1\}^{n \times d}$ be the
  set of possible $X$ that satisfy the inner product condition
  $\abs{\inner{x_i, y_i}} > a \sqrt{d}$ for all rows $i \in [n]$
  with $b_i = 1$.  Each row with $b_i = 1$ has at most $2^{(1-\Omega(a^2))d}$ values of
  $x_i$ in the support, so
  \[
    \abs{S_{Y,b}} \leq 2^{nd - \Omega(a^2 \norm{b}_1 d)}.
  \]
  We have $\E[\norm{b}_1] \geq \frac{n}{2}$, so
  \[
    H(X \mid Y) \leq H(X \mid Y, b) + H(b) \leq (\E_{Y, b} \log \abs{S_{Y,b}}) + n \leq (1 - \Omega(\frac{1}{2} a^2))nd  + n
  \]
  so
  \[
    I(X; Y) = H(X) - H(X \mid Y) \geq \Omega(a^2 nd) - n.
  \]
\end{proof}


\begin{claim}\label{claim:maxa}
  Let $A, B > 0$.  Then
  \[
    A a^2 + B a b \leq \frac{a^2 + b^2}{2}(A + \sqrt{A^2 + B^2}),
  \]
  with equality if $\frac{a^2}{a^2 + b^2} = \frac{1 + \sqrt{\frac{A^2}{A^2 + B^2}}}{2} $.
\end{claim}
\begin{proof}
  Just ask a computer.  By hand, though: the equations are homogeneous,
  so WLOG we can assume $a^2 + b^2 = 1$.  We then maximize over
  $a \in [0, 1]$.  Taking the derivative, the maximum is achieved when
  \[
    2 A a + B (\sqrt{1-a^2}  - \frac{a^2}{\sqrt{1-a^2}}) = 0
  \]
  or
  \begin{align*}
    2Aa \sqrt{1-a^2} &= B (2a^2 - 1)\\
    4A^2a^2 (1-a^2) &= B^2 (4a^4 - 4a^2 + 1)\\
    a^4 (4B^2 + 4A^2) - a^2(4A^2 + 4B^2)  + B^2 &= 0\\
    a^2 = \frac{1 \pm \sqrt{\frac{A^2}{A^2 + B^2}}}{2}
  \end{align*}
  the first squaring preserved equality only when $a^2 \geq \frac{1}{2}$, so the optimum is at
  \[
    a^2 = \frac{1 + \sqrt{\frac{A^2}{A^2 + B^2}}}{2}.
  \]
  Then
  \begin{align*}
    A a^2 + B a \sqrt{1-a^2} &= A \frac{1 + \sqrt{\frac{A^2}{A^2 + B^2}}}{2} + B  \sqrt{\frac{1 + \sqrt{\frac{A^2}{A^2 + B^2}}}{2} \frac{1 - \sqrt{\frac{A^2}{A^2 + B^2}}}{2}}\\
                             &= A \frac{1 + \sqrt{\frac{A^2}{A^2 + B^2}}}{2} + B  \sqrt{\frac{\frac{B^2}{A^2 + B^2}}{4}}\\
                             &= \frac{1}{2}(A + \sqrt{A^2 + B^2}).
  \end{align*}
\end{proof}

\begin{lemma}[Laurent-Massart Bounds\cite{LaurentMassart}]
    \label{lem:chi_squared_concentration}
    Let $v \sim \mathcal{N}(0, I_n)$. For any $t > 0$, 
    \[
    \Pr[\norm{v}^2 - n \ge 2 \sqrt{nt}+ 2t] \le e^{-t},
    \]
        \[\Pr[\norm{v}^2 - n \leq -2\sqrt{nt}] \leq e^{-t}.\]
\end{lemma}

\begin{lemma}[\cite{MV10}, see also \cite{cryptoeprint:2015/1128}]
    \label{lem:spherical_cap_area}
    Consider a $d$-dimensional unit sphere $S_d$.
    Let $C_h$ be a spherical cap on $S_d$ with height $h < 1$, i.e., \[
        C_h := \{v \in S_d \mid \inner{u, v} \ge 1 - h\}
    \] for some $u \in S_d$. Then the ratio of the area of $C_h$ to the area of $S_d$ is given by $ d^{\Theta(1)} \cdot (2h - h^2)^{d/2}$.
\end{lemma}

\section{Lower Bound for Linear Sketching}
\label{app:lw}
When establishing lower bounds for approximating operator norms using linear sketching, Li and Woodruff~\cite{LiWoodruff} constructed a lower bound instance with the following properties:

\begin{lemma}
    \label{lem:LiWoodruff}
    For any $\alpha > 1.01$, there exist two distributions $\D_1$ and $\D_2$ over $\R^{d \times d}$ and $s > 0$ such that 
    \begin{enumerate}
        \item For $X \sim \D_1$, $\|X\|_2 > \sqrt{\alpha} s$ with 0.99 probability.
        \item For $X \sim \D_2$. $\|X\|_2 < s$ with 0.99 probability.
        \item For $X \sim \D_1$, the spectral gap $\lambda_1(X^TX) / \lambda_2(X^TX)$ is at least $\alpha$ with $0.99$ probability.
        \item  Let $\mathcal{L}_1$ and $\mathcal{L}_2$ be the corresponding distribution of the linear sketch of dimension $k$ on $\mathcal{D}_1$ and $\mathcal{D}_2$. Then $d_{TV}(\mathcal{L}_1, \mathcal{L}_2) < 0.1$ whenever $k \le o(d^2 / \alpha^2)$.
    \end{enumerate}
\end{lemma}

This implies a space lower bound for PCA using linear sketching. We present the Johnson-Lindenstrauss lemma first.

\begin{lemma}[Johnson-Lindenstrauss Lemma~\cite{JL}]
\label{lem:JL}
For any positive integer $d$ and $\varepsilon, \delta \in (0,1)$, there exists a distribution $\mathcal{S}$ over $\R^{m \times d}$ where $m = \Theta\left(\varepsilon^{-2} \log \frac{1}{\delta}\right)$ such that for every $x \in \mathbb{R}^d$,
\[
\Pr_{A \sim \mathcal{S}} \left[ \left| \|Ax\|_2^2 - \|x\|_2^2 \right| \leq \varepsilon \|x\|_2^2 \right] \geq 1 - \delta.
\]
\end{lemma}

Using these lemmas, we can prove the following lower bound, which implies any sketching algorithm for adversarial streaming PCA needs at least $\Omega(d^2 / R^2)$ bits of space.

\begin{theorem}
For all linear sketching algorithms, $0.1$-approximate PCA on streams with spectral gap $R = o(\sqrt{d})$ requires sketches of dimension $\Omega(d^2 / R^2)$,
\end{theorem}

\begin{proof}
    Let $\D_1$ and $\D_2$ be the distributions described in \cref{lem:LiWoodruff} with $\alpha = R$. Suppose there exists a linear sketching algorithm that 0.1-approximates PCA using $o(d^2 / R^2)$ space with success probability $0.99$. We will show that this leads to a contradiction by constructing a linear sketch of dimension $o(d^2 / R^2)$ that distinguishes $\D_1$ and $\D_2$ with 0.9 probability whenever $4 \le R \le o(\sqrt{d})$.

    Let $\mathcal{S}$ be the distribution in \cref{lem:JL} with parameters $\delta = \eps = 0.01$, and $\mathcal{S}$ is a distribution over $\mathbb{R}^{O(1) \times d}$. Let $s$ be the corresponding parameter for $\D_1$ and $\D_2$ in \cref{lem:LiWoodruff}. Our algorithm proceeds as follows: Given a matrix $X$, run the PCA approximation algorithm, which is a linear sketching of dimension $o(d^2 / R^2)$, to obtain an approximation $\wt{v}$. In parallel, sample $A \sim \mathcal{S}$ and compute $AX$, which is a matrix of dimension $O(1) \times d$; that is, it is a linear sketch with $O(d) = o(d^2/R^2)$ dimensions. Suppose $\|AX\wt{v}\|_2 > 1.1s$, output that $X$ is from $\mathcal{D}_1$; otherwise, output that $X$ is from $\mathcal{D}_2$.

      We first show that for $X \sim \mathcal{D}_1$, $ \|AX\wt{v}\|_2 > 1.1 s$ with 0.9 probability. Let $v^*$ be the true principal component of $X$ in the direction that $\inner{\wt{v}, v^*} \ge 0$. By a union bound, we have that with probability at least $0.96$, the following events happen simultaneously:
      \begin{enumerate}
          \item $X$ has a spectral gap of $R$.
          \item The PCA approximation $\wt{v}$ satisfies $\sin^2(\widetilde{v}, v^*) \leq 0.1$.
          \item $\|AX\widetilde{v}\|_2 \geq 0.99 \|X\wt{v}\|_2$.
          \item $\|X\|_2 > \sqrt{R}s$.
      \end{enumerate}
      When all of these hold, we can prove that $\|AX\wt{v}\|_2 > 0.6\sqrt{R}s$. We have \[
      \|AX\widetilde{v}\|_2 \ge 0.99 \|X\wt{v}\|_2 \ge 0.99 (\|Xv^*\|_2 - \|X(v^* - \wt{v})\|_2) \geq 0.99 (\|X\|_2 - \|v^* - \wt{v}\|_2 \|X\|_2 ).
      \]
      Since $\wt{v}$ and $v^*$ are unit vectors, $\sin^2(\widetilde{v}, v^*) \leq 0.1$ implies that $\|v^* - \wt{v}\|_2 = \sqrt{2 - 2 \cos(\widetilde{v}, v^*)} \le 0.35$. Thus, \[
        \|AX\widetilde{v}\|_2 \ge 0.6 \|X\|_2 > 0.6\sqrt{R}s.
      \]
      Therefore, $\|AX\widetilde{v}\|_2 > 1.1 s$ with probability at least $0.9$ whenever $R \geq 4$.

    Next, we show that for $X \sim \mathcal{D}_2$, $\|AX\wt{v}\|_2 \le 1.1 s$ with 0.9 probability. Again, by a union bound, we have with at least 0.9 probability, \[
        \|AX\wt{v}\|_2 \le 1.01 \|X\wt{v}\|_2 \le 1.01 \|X\|_2 \le 1.1s.
    \]
    
    This proves that our sketching algorithm distinguishes $\D_1$ and $\D_2$ with probability at least 0.9, contradicting \cref{lem:LiWoodruff}, and therefore proves the theorem. 
\end{proof}
\end{document}